\documentclass{lmcs} 
\pdfoutput=1

\usepackage[utf8]{inputenc}
\usepackage{lastpage}
\lmcsdoi{19}{4}{29}
\lmcsheading{}{\pageref{LastPage}}{}{}%
{Nov.~01,~2022}{Dec.~18,~2023}{}


\usepackage{amssymb}
\usepackage{amsmath}
\usepackage{subfig}
\usepackage{listings}
\usepackage{acronym}
\usepackage{adjustbox}
\def\doi#1{\href{https://doi.org/\detokenize{#1}}{\url{https://doi.org/\detokenize{#1}}}}
\usepackage{graphicx}
\usepackage{nicefrac}
\usepackage{algorithm}
\usepackage[noend]{algpseudocode}

\usepackage{tikz}
\usetikzlibrary{arrows,shapes.geometric,positioning,automata,patterns}
\usetikzlibrary{trees,decorations.pathmorphing,calc}
\usetikzlibrary{intersections}
\usetikzlibrary{decorations.pathreplacing,external,backgrounds}
\usepackage{etoolbox} 

\definecolor{ddarkgreen}{rgb}{0,0.5,0}

\lstdefinelanguage{Protelis}{
    language=Java,
    morecomment=[l]{//},
    morecomment=[s]{/*}{*/},
    keywordstyle=[2]\color{blue},
    keywordstyle=[3]\color{Mahogany},
    keywords={def, else, import, mux, sumHood, PlusSelf, let, rep, env, self, pi, if, public, NaN, false, true, minHood, nbr, it, share},
    keywords=[2]{POSITIVE_INFINITY, map, reduce, sin, cos, foldMin, min, get, set, gossip, distanceTo},
    keywords=[3]{nbrRange, nbrVersor, foldUnion, foldSum},
    frame=no,
    breaklines=true,
    captionpos=b,
    keepspaces=true,  numbers=left,
    backgroundcolor=\color{white},    
    basicstyle=\ttfamily\small,
    keywordstyle=\bfseries\color{purple},
    emphstyle={\bfseries\color{blue!40!black}},
    stringstyle=\ttfamily\color{orange},
    commentstyle=\ttfamily\color{green!40!black},
      numbers=left,
  xleftmargin=0.7cm, 
  escapechar=\%
}
\tikzset{-,
  ev/.style={circle,draw, inner sep=1.8pt, minimum size=2pt, outer sep=0.5mm},
  evlink/.style={line width=0.4mm},
  point/.style={circle,fill=black, inner sep=1.8pt, minimum size=2pt, outer sep=0.5mm},
  pointp/.style={rectangle,fill=black, inner sep=2pt, minimum size=2pt, outer sep=0.5mm},
  event/.style={circle,fill=red, inner sep=3pt, minimum size=2pt, outer sep=0.5mm},
  past/.style={event, fill=ddarkgreen},
  present/.style={event,fill=white,fill opacity=0.0,text opacity=1.0,draw,line width=0.4mm},
  future/.style={event,fill=blue},
  locality/.style={circle,fill=white, inner sep=2pt, minimum size=2pt, outer sep=0.5mm},
  locality1/.style={locality,fill=orange!50},
  locality2/.style={locality,fill=gray!50},
  locality3/.style={locality,fill=blue!40},
  locality4/.style={locality,fill=black!30!green},
  locality5/.style={locality,fill=purple!50},
  d/.style={circle,fill=gray, inner sep=2pt, minimum size=2pt},
  da/.style={circle,draw,label={[left]$1$}}, 
  db/.style={circle,draw,label={[left]$2$}}, 
  dc/.style={circle,draw,label={[left]$3$}}, 
  lbl/.style={font=\footnotesize},
  state/.style={circle,draw}, 
  statein/.style={circle,draw,line width={2pt}},
  stateout/.style={circle,draw,line width={1pt},fill=black!16},
  leadstov/.style={draw,decorate,decoration={snake, amplitude=1.95mm, post=lineto, post length=2mm, segment length=1cm, pre length=0.4cm},->,>=stealth'}, 
  leadsto/.style={draw,decorate,decoration={snake, amplitude=0.35mm, post=lineto, post length=2mm, segment length=1.3mm},->,>=stealth'}, 
  leadstopast/.style={leadsto,draw=ddarkgreen},
  leadstopresent/.style={leadsto,draw=black},
  leadstofuture/.style={leadsto,draw=blue},
  frastagliato/.style={draw,decorate,decoration={snake, amplitude=0\imgfactormm, post=lineto, post length=0mm, segment length=2.3mm},->,>=stealth'},
  appgraph/.style={rectangle,draw,minimum size=1.5cm},
  appgraphc/.style={ellipse,draw,minimum size=1.5cm},
  program/.style={rectangle,draw,radius=0.1cm},
  scheduler/.style={rounded rectangle,draw,radius=0.1cm,font=\small},
  evpast/.style={green},
  evpresent/.style={red},
  evfuture/.style={blue},
  evconc/.style={gray},
  dev/.style={circle,draw, inner sep=1.8pt, minimum size=2pt, outer sep=0.5mm}
}

\usepackage[inline]{enumitem}
\newlist{inlinelist}{enumerate*}{1}
\setlist[inlinelist]{label=\emph{(\roman*)}}
\usepackage[T1]{fontenc}

\usepackage{hyperref}
\usepackage{cleveref}

\newcommand{\addrev}[1]{#1}

\acrodef{hvac}[HVAC]{Heating, Ventilation, and Air Conditioning}
\acrodef{wsn}[WSN]{Wireless Sensor Network}
\acrodef{wsan}[WSAN]{Wireless Sensor and Actuator Network}
\acrodef{qos}[QoS]{Quality of Service}
\acrodef{cps}[CPS]{Cyber-Physical System}
\acrodef{tdma}[TDMA]{Time-Division Multiple Access}
\acrodef{iot}[IoT]{Internet-of-Things}
\acrodef{poi}[PoI]{Point of Interest}
\acrodefplural{poi}[PoIs]{Points of Interest}

\newcommand{\sampling}{aggregate sampling}
\newcommand{\Sampling}{\expandafter\MakeUppercase\sampling}
\newcommand{\sampler}{aggregate sampler}

\newcommand{\esResizeFactor}{1.0}

\newcounter{mylabelcounter}

\makeatletter
\newcommand{\labeltext}[2]{%
\label{#1-place}{
#2\refstepcounter{mylabelcounter}%
\immediate\write\@auxout{%
  \string\newlabel{#1}{{1}{\thepage}{{\unexpanded{#2}}}{mylabelcounter.\number\value{mylabelcounter}}{}}%
}%
}
}
\makeatother

\theoremstyle{plain} 

\def\imgfactor{0.48}

\newcommand{\revA}[1]{#1} 
\newcommand{\revT}[1]{\revA{#1}} 
\newcommand{\revB}[1]{#1} 

\sloppypar

\begin{document}


\newcommand{\tuple}[1]{\ensuremath{\langle #1 \rangle}}
\newcommand{\evl}[2]{\ensuremath{\epsilon_{#1}^{[#2]}}}

\newcommand{\deviceId}{\delta}
\newcommand{\eventId}{\epsilon}
\newcommand{\senstate}{\ensuremath{x}} 
\newcommand{\SenStates}{\ensuremath{X}}
\newcommand{\deviceS}{\Delta}
\newcommand{\eventS}{\ensuremath{E}}
\newcommand{\aEventS}{\ensuremath{\mathbf{E}}}
\newcommand{\E}{\aEventS}
\newcommand{\EI}{\aEventS}
\newcommand{\ES}{\ensuremath{\mathcal{E}}}
\renewcommand{\L}{\ensuremath{\mathcal{L}}}
\newcommand{\I}{\ensuremath{\mathcal{I}}}
\newcommand{\ST}{\ensuremath{\mathit{\mathbf{ST}}}}
\newcommand{\STR}{\ensuremath{\mathcal{R}}}
\newcommand{\past}{\ensuremath{\mathit{Past}}}

\newcommand{\SpatialSec}{\ensuremath{\Sigma}}
\newcommand{\Sample}{\ensuremath{\sigma}}
\newcommand{\ValueSet}{\ensuremath{\mathbb{V}}}
\newcommand{\FieldSet}{\ensuremath{\mathbb{F}}}
\newcommand{\field}{\ensuremath{f}}
\newcommand{\Program}{\ensuremath{P}}
\newcommand{\FieldComp}{\ensuremath{\Phi}}
\newcommand{\samplingErr}{\ensuremath{\eta}}

\newcommand{\tfuture}[1]{\ensuremath{T_{#1}^{+}}}
\newcommand{\tpast}[1]{\ensuremath{T_{#1}^{-}}}
\newcommand{\tpastof}[2]{\ensuremath{T_{#1}^{-}(#2)}}
\newcommand{\trestrict}[2]{#1|_{#2}}
\newcommand{\tfutureof}[2]{\ensuremath{\trestrict{#1}{\tfuture{#2}}}}

\newcommand{\Regions}{\ensuremath{regions}}
\newcommand{\Img}{\ensuremath{Img}}
\newcommand{\neigh}{\rightsquigarrow}
\newcommand{\sneigh}{\rightarrowtail}
\newcommand{\devof}{d}
\newcommand{\sensof}{p} 

\newcommand{\receivers}[0]{\ensuremath{\mathit{recvs}}}

\newcommand{\fmetric}[0]{\ensuremath{\mu}}
\newcommand{\errth}[0]{\ensuremath{\xi}}

\newcommand{\cp}[1]{{\left( #1 \right)}}
\newcommand{\qp}[1]{{\left[ #1 \right]}}
\newcommand{\ap}[1]{{\langle #1 \rangle}}
\newcommand{\bp}[1]{{\left\lbrace #1 \right\rbrace}}
\newcommand{\vp}[1]{{\left\lvert #1 \right\rvert}}
\newcommand{\floor}[1]{{\left\lfloor #1 \right\rfloor}}
\newcommand{\ceil}[1]{{\left\lceil #1 \right\rceil}}

\newcommand{\BoolType}{\ensuremath{\mathit{Bool}}}
\newcommand{\MetricType}{\ensuremath{\mathit{Metric}}}
\newcommand{\FloatType}{\ensuremath{\mathit{Float}}}

\newcommand{\True}{\ensuremath{\top}}
\newcommand{\False}{\ensuremath{\bot}}
\newcommand{\R}{\ensuremath{\mathbb{R}}}

\newcommand{\MsgSet}{\ensuremath{\mathcal{M}}}

\title[Space-fluid Adaptive Sampling by Self-Organisation]{Space-fluid Adaptive Sampling by Self-Organisation}
\titlecomment{{\lsuper*}
This article is an extended version of the conference paper~\cite{DBLP:conf/coordination/CasadeiMPVZ22} presented at COORDINATION'22.
}

\author[R.~Casadei]{Roberto Casadei\lmcsorcid{0000-0001-9149-949X}}[a]	

\author[S.~Mariani]{Stefano Mariani\lmcsorcid{0000-0001-8921-8150}}[b]	

\author[D.~Pianini]{Danilo Pianini\lmcsorcid{0000-0002-8392-5409}}[a]	
\author[M.~Viroli]{\texorpdfstring{\\}{}Mirko Viroli\lmcsorcid{0000-0003-2702-5702}}[a]	

\author[F.~Zambonelli]{Franco Zambonelli\lmcsorcid{0000-0002-6837-8806}}[b]	

\address{Alma Mater Studiorum---Università di Bologna, Via dell'Università, 50, Cesena (FC), Italy}	
\email{roby.casadei@unibo.it, danilo.pianini@unibo.it, mirko.viroli@unibo.it}  

\address{Università di Modena e Reggio Emilia, Via Giovanni Amendola, 2, Reggio Emilia (RE), Italy}	
\email{stefano.mariani@unimore.it, franco.zambonelli@unimore.it}  





\begin{abstract}
  A recurrent task in coordinated systems 
  is managing (estimating, predicting, or controlling) signals that vary in space, such as distributed sensed data or computation outcomes.
 Especially in large-scale settings, 
  the problem can be addressed 
  through decentralised and situated computing systems:
  nodes can locally sense, process, and act upon signals,  
  and coordinate with neighbours
  to implement collective strategies.
 Accordingly, in this work 
  we devise distributed coordination strategies
  for the estimation of a spatial phenomenon
  through collaborative adaptive sampling.
 Our design is based on the idea of dynamically partitioning space into regions 
  that compete and grow/shrink to provide accurate aggregate sampling. 
 Such regions hence define a sort of virtualised space that is ``fluid'', 
  since its structure adapts in response to pressure forces exerted by the underlying phenomenon.
 We provide an adaptive sampling algorithm in the field-based coordination framework\addrev{,
 and prove it is self-stabilising and locally optimal}. 
 Finally, we verify by simulation 
  that the proposed algorithm effectively carries out a spatially adaptive sampling 
  \addrev{while maintaining a tuneable trade-off between accuracy and efficiency}.
\end{abstract}

\maketitle

\section{Introduction}

A significant problem in computer systems engineering 
 is dealing with phenomena that vary in space:
 for instance, their estimation, prediction, and control.
Concrete related application examples include: the monitoring of waste in urban areas to improve waste gathering strategies~\cite{imran2019smart-waste};
 the estimation of pollution in a geographical area, 
  for alerting or mitigation-aimed response purposes~\cite{casadei2020pulverization};
  the sensing of the temperature in a large building,
  to support the synthesis of control policies for the \emph{\ac{hvac}} system~\cite{Manjarres2017hvac-predictive-control}.
The general solution for addressing this kind of problem 
 consists of deploying a set of sensors and actuators
 in space, and building a distributed system
 that processes gathered data and possibly determines a suitable actuation in response~\cite{DBLP:journals/percom/WuKT11}.
In many settings,
 the computational activity
 can (or has to) be performed \emph{in-network}~\cite{DBLP:journals/wc/FasoloRWZ07} 
 in a \emph{decentralised} way:
 in such systems, nodes locally sense, process, and act upon 
 the environment, 
 and coordinate with \emph{neighbour} nodes
 to collectively self-organise their activity.
However, in general 
 there exists a trade-off
 between performance and efficiency,
 that suggests concentrating the activities on few nodes,
 or to endow systems with the capability of autonomously 
 adapt the \emph{granularity} of computation~\cite{DBLP:journals/tcyb/YaoVP13}.

In this work,
 we focus on sampling signals that vary in space.
Specifically, 
 we would like to sample a spatially distributed signal
 through device coordination and self-organisation
 such that the samples accurately reflect the original signal
 and the least amount of resources is used to do so.
In particular, we push forward a vision
 of space-fluid computations,
 namely computations that are fluid, i.e.\ change seamlessly, in space and -- like fluids -- adapt in response to pressure forces 
 exerted by the underlying phenomenon.
We reify the vision through an algorithm
 that handles the shape and lifetime of
 leader-based ``regional processes'' (cf.~\cite{FGCS2020-scr}),
 growing/shrinking as needed to sample 
 a phenomenon of interest
 with a (locally) maximum level of accuracy and minimum resource usage.
For instance,
 we would like to sample more densely those regions of space
 where the spatial phenomenon under observation
 has high variance,
 to better reflect its spatial dynamics.
 On the contrary,
 in regions where variance is low,
 we would like to sample the phenomenon more sparsely 
 to, e.g., save energy, communication bandwidth, etc.
 while preserving the same level of accuracy.

Accordingly, 
we consider the \emph{field-based coordination} framework of \emph{aggregate computing}~\cite{BealIEEEComputer2015,JLAMP2019},
 which has proven to be effective in 
 modelling and programming self-organising behaviour
 in situated networks of devices interacting asynchronously.
On top of it, we devise a solution
that we call \emph{\sampling{}},
inspired by the approaches of self-stabilisation~\cite{TOMACS2018} and density-independence~\cite{TAAS2017},
that maps an input field representing a signal to be sampled into a regional partition field where each region provides a single sample;
then, we characterise the \sampling{} error based on a distance defined between stable snapshots
of regional partition fields,
and propose that an \emph{effective} \sampling{} is one that is locally optimal w.r.t.\ an error threshold,
meaning that the regional partition cannot be improved simply by merging regions.
In summary, we provide the following contributions:
\begin{itemize}
\item we define a model for distributed collaborative adaptive sampling and characterise the corresponding problem in the field-based coordination framework;
\item we implement an algorithmic solution to the problem that leverages self-organisation patterns like gradients~\cite{DBLP:journals/nc/Fernandez-MarquezSMVA13,TOMACS2018} and coordination regions~\cite{FGCS2020-scr};
\item \addrev{we prove this algorithm to self-stabilise, and to actually provide an effective sampling according to a definition of ``locally optimal regional partition'';} 
\item we experimentally validate the algorithm to verify interesting trade-offs between sparseness of the sampling and its error.
\end{itemize}
\addrev{
This manuscript is an extended version of the conference paper~\cite{DBLP:conf/coordination/CasadeiMPVZ22}, providing 
\begin{inlinelist} 
\item a more extensive and detailed coverage of related work;
\item more examples, clarifications, and details regarding the formal model;
\item a discussion of the source code of the aggregate computing implementation; and
\item proofs of self-stabilisation and local optimality of the proposed algorithm.
\end{inlinelist}
}

The rest of the paper is organised as follows.
\Cref{s:background} covers motivation and related work.
\Cref{s:problem} provides a model for distributed sampling
 and the problem statement.
\Cref{s:solution} describes an algorithmic solution to the problem of sampling a distributed signal using the framework of aggregate computing.
\Cref{s:eval} performs an experimental validation of the proposed approach.
Finally, \Cref{s:conc} provides conclusive thoughts and delineates directions for further research.

\section{Motivation and Related Work}
\label{s:background}

\subsection{Motivations, Goal, and Applications}
\label{ss:goal_apps}

Consider a \ac{wsn}
of any topology,
\emph{statically} (i.e.\ design-time, no mobility) deployed across a geographical area
to monitor a spatially-distributed phenomenon, such as, for instance, air quality,
as depicted in \Cref{fig:airqualitymap}.
\begin{figure}[t]
	\centering
	\includegraphics[width=\columnwidth]{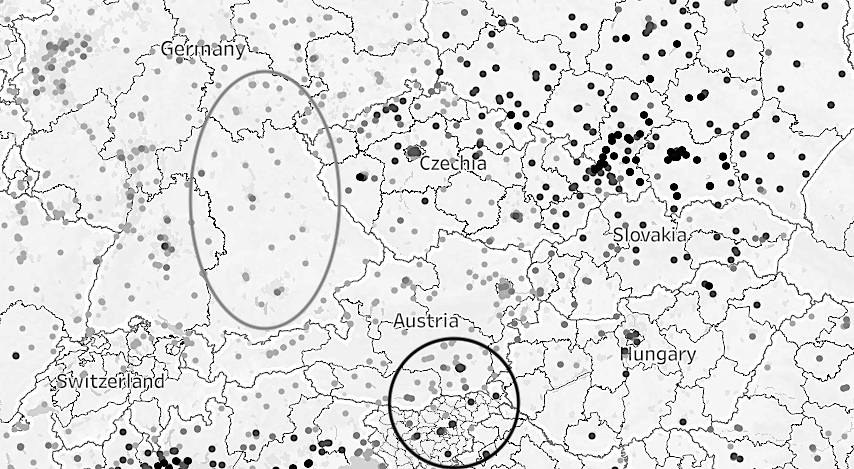}
	\caption{
		Air quality statistics map taken from \url{https://archive.ph/dMJO2}.
		There are areas where the underlying phenomenon does not vary significantly in space
		(light-grey oval),
		hence sampling could be made sparser
		with tolerable loss of accuracy.
		In others (darker circle), variance is high, 
		requiring a more detailed spatial sampling.
	}
	\label{fig:airqualitymap}
\end{figure}
We want to \emph{dynamically} (at run-time) and \emph{adaptively} (depending on the phenomenon itself) find a 
sparse set of \emph{samplers}, i.e., devices responsible for providing sensing data regarding some underlying phenomenon.
We want the selection of samplers 
 to depend on both the \emph{spatial distribution} of devices
 and the input phenomenon.
Therefore, the idea is that each sampler 
 is responsible for an exclusive \emph{spatial sampling region} 
 that may include several other devices, i.e., a partition of the system/environment.
Moreover, 
we want to determine a 
partitioning of the connected sampling devices that
minimises the number and maximises the size of sampling regions,
while preserving as much as possible the underlying information.
Hence, in areas with low variance amongst spatially distributed samples' values,
we want our regions to be larger,
as many samples will report similar values, and hence one sample nicely represents all.
Conversely, in areas with high spatial variance between samples,
smaller regions are necessary
as even proximal samples may have very different values, hence more samples are required to accurately represent the phenomenon.
We also consider \emph{large-scale deployments} 
 with hundreds or thousands of devices.

Accordingly, we aim at designing a \emph{decentralised algorithm} that can:
\begin{itemize}
    \item dynamically (at run-time, continuously) partition a set of sampling devices into sampling regions; 
    \item consider phenomenon-specific metrics (e.g.\ variance of the sensors' readings) for deciding how to (re)compute partitions;
    \item both up-scale and down-scale sampling depending on such metrics;
    \item do so in a fully distributed way, based solely on local interactions (i.e., within a 1-hop neighbourhood).
\end{itemize}
%
%
A similar algorithm can provide benefits across multiple application domains, 
as also witnessed by the below described related literature. 
All forms of environmental monitoring can greatly benefit, for instance, as many phenomena in such a domain are inherently spatially distributed: 
air quality, water pollution, soil radiation, landslide monitoring, crop growth, and so on~\cite{Cox99RiskAnalysis,YOO2020117091,DBLP:conf/lcn/ZamanGMB14,DBLP:journals/ejwcn/MysorewalaCP12}.
Another field of application is ecological monitoring, such as geolocation of wolf packs or other animals moving habits~\cite{DBLP:journals/sac/PeyrardSSBN13}.
But in general, any application whose goal is to monitoring any measurable phenomenon, with unknown or uncertain spatial dynamics, may find benefits in our proposed approach.

\subsection{Related Work on Adaptive Spatial Sampling}
\label{ss:related}

There are several approaches in the literature that attempt to solve this and similar problems,
with heterogeneous techniques, that we collectively refer to under the umbrella term of ``adaptive spatial sampling''.
Amongst these, some~\cite{DBLP:journals/corr/abs-1811-01303,doi:10.1080/01621459.1990.10474975} are restricted to the so-called \emph{sampling design} problem, 
that is, their concern is to deliver either 
design-time decision support about where to deploy sensor devices 
or analytically devise out the best sampling algorithm 
given domain expertise or infrastructural requirements 
    (most often, residual energy management).
Our approach is not directly comparable to these, 
as we are concerned with run-time adaptation 
of the sampling process 
based on domain-specific properties 
    (e.g.\ sampling variance, that depends on the phenomenon under observation).
Others~\cite{DBLP:conf/iros/RahimiHKSE05,DBLP:conf/rss/GargA14,DBLP:journals/corr/abs-2105-10018,DBLP:journals/iet-wss/HamoudaP11,DBLP:conf/cdc/GrahamC09} 
more closely pursue our goal, 
but assume mobile sensing devices (e.g.\ robots), 
hence are concerned with how to move them at run-time to optimise some desired metric 
    (e.g.\ sampling accuracy).
On the contrary, we assume static sensor devices 
that have been already deployed on a target area, 
without any prior knowledge of the actual spatial distribution of the phenomenon to observe.
Finally, a great deal of related research contributions 
have the fundamental difference of adapting the sampling process to \emph{domain-agnostic}, infrastructural properties 
such as residual energy, distance amongst devices, bandwidth consumption, 
rather than on the specific spatial distribution (and dynamics) of the phenomenon under monitoring~\cite{DBLP:conf/infocom/YounisF04,DBLP:journals/adhoc/MhatreR04,DBLP:journals/cn/BandyopadhyayC04,DBLP:conf/infocom/BandyopadhyayC03,DBLP:journals/wc/SohrabiGAP00,DBLP:journals/pieee/StankovicALSH03,DBLP:conf/infocom/BhardwajC02,DBLP:journals/tac/OgrenFL04}.

Narrowing down the research landscape just overviewed, 
we now describe and compare in more detail the approaches to spatial adaptive sampling most similar to ours, 
emphasising the most notable differences.

In reference~\cite{DBLP:conf/mass/VirrankoskiS05} 
the Topology Adaptive Spatial Clustering algorithm 
is presented. 
It is a distributed algorithm
that partitions a \ac{wsn} into a set of disjoint sampling clusters,
with no prior knowledge of cluster number or size (like ours), 
by encoding geographical distance, connectivity, and deployment density information 
in a single measure upon which leader election (for cluster heads) happens. 
The goal is to group together nodes in proximity 
and within regions of similar deployment density, 
to improve efficiency of data aggregation and compression.
Besides the focus on efficiency, 
that is only half the story in our approach 
    (the other being accuracy, hence the trade-off), 
two are the fundamental differences our approach has with respect to this: 
first, in our case adaptation is domain dependent, 
in the sense that it depends on some property (e.g.\ variance) of the phenomenon under observation, 
not on infrastructural properties; 
second, we do also consider over-sampling when variation is high, 
whereas reference~\cite{DBLP:conf/mass/VirrankoskiS05} only considers under-sampling where measures are redundant.

In reference~\cite{DBLP:conf/aina/LiuXZWL13}, 
another distributed approach to spatial adaptive sampling is presented.
It is based on the assumption that neighbouring sampling nodes usually have similar readings 
    (high spatial correlation), 
hence can be grouped in a cluster to improve energy consumption.
The proposed algorithm uses such spatial correlation, 
two application-specific threshold parameters 
    (error tolerance and correlation range), 
and residual energy to elect cluster heads, 
with the goal of minimising the number of clusters, 
and the variance of their size.
The introduction of the two application-dependant parameters 
makes this proposal closer to ours, 
as they could be used to steer the adaptation toward domain-specific aspects, 
to some extent.
However, most of the calculations still rely on infrastructural properties 
rather than measures about the phenomenon of interest, 
and the focus is, once again, on energy saving, thus,
authors do not consider over-sampling but solely under-sampling.

Reference~\cite{DBLP:conf/ccnc/LinM07} 
shares with us the interest in performing adaptation based on 
the information observed by sensors, 
rather than on network, energy, or other infrastructural aspects. 
However, they consider a special case 
where the clusters must correspond to pre-determined ``sources of interest'', 
such as physical objects/devices in the environment. 
Moreover, clusters are formed with the secondary objective of being balanced, 
whereas we allow them to be of different sizes and shapes 
    (and even encourage to be so depending on the spatial distribution of the phenomenon of interest).
    
Reference~\cite{DBLP:conf/icac/LeeVP11} 
proposes SILENCE, a distributed, space-time adaptive sampling algorithm 
based on (space-time) correlation of measured values.
The foremost goal pursued is efficiency: 
SILENCE in fact strives to minimise communication and processing overhead, 
by minimising sampling redundancy 
and also adapting (i.e.\ slowing down) the scheduling speed of sensor devices.
However, yet again the basic assumption, 
and focus of the approach, 
is the case where spatial correlation of sampled values is high; 
whereas we explicitly focus on the opposite situation 
    (while still providing a working solution even in the case considered by SILENCE).
Furthermore, also SILENCE only considers down-sampling.

Finally, the ASample algorithm~\cite{DBLP:conf/sutc/SzczytowskiKS10} 
is the one approach most similar to ours, 
not in the techniques exploited, 
but in pursuing domain-driven adaptation, 
in doing so in a fully-distributed way, 
and in considering the opportunity to over-sample, too.
In particular, ASample builds a Voronoi tessellation of the area 
where the \ac{wsn} is deployed,
in a fully distributed way 
by considering only neighbourhood information, 
instead of the whole topology.
Such a tessellation considers a desired sampling accuracy, 
specified at the application level: 
while a given Voronoi region is within the accuracy bound, 
it keeps expanding; 
on the contrary, whenever the accuracy constraint is violated, 
a virtual centroid of a novel Voronoi region is spawned, 
with a value that is obtained through interpolation of the neighbouring regions.
This is an important aspect to consider, 
as it introduces synthetic data, which is something we avoid 
as we only increase sampling granularity if there are actual devices available in the target area.
Moreover, there is an assumption underlying the ASample approach 
that does not hold in general, 
and, specifically, it is in contrast with our intended goal 
    (obtaining potentially irregular clusters to reflect the irregular spatial distribution of the observed phenomenon): 
it is assumed that the smaller the area covered by the Voronoi region, 
the less representative the samples drawn are, 
hence the smaller the impact on the global sampling accuracy.
In our targeted scenarios, the opposite could be true, too: 
smaller regions represent sharp variance of measurements across space, 
and more accurately represent the irregularities of the underlying phenomenon.

\section{Distributed Aggregate Sampling: Model}
\label{s:problem}

In order to define the problem and characterise our approach,
 we leverage the event structure framework~\cite{DBLP:journals/tcs/NielsenPW81,ABDV-COORDINATION2018},
  which provides a general model of situated computations.
\addrev{
Within this formal framework,
 in this section we 
 describe
 the computational model (\Cref{ssec:comp-model}),
 self-stabilisation as a desired property of solutions in this model (\Cref{ssec:self-stab}),
 and the spatial sampling problem that we tackle in this manuscript (\Cref{ssec:problem-def}).
}
\addrev{
The computational model introduces the necessary terminology to understand both the problem formulation and the solution we propose in \Cref{s:solution}.
Introducing the self-stabilisation property is required to be able to evaluate such a solution in its effort to adapt to the dynamics of the phenomenon of interest, both formally as done in \Cref{ssec:formal-results} and practically as done later in \Cref{s:eval}.
}

\subsection{Computational Model}
\label{ssec:comp-model}

We consider a computational model where a set of \emph{devices} \revA{(typically comprising \emph{sensors} and \emph{actuators} to perceive and act upon the environment)}  compute at discrete steps called \emph{computation rounds} 
and interact with \emph{neighbour devices} by exchanging messages.
\revA{Executions of such systems can be modelled through \emph{event structures} \revB{\cite{DBLP:journals/tcs/NielsenPW81,DBLP:journals/ijpp/Pratt86}} as in \cite{TOCL2019,ABDV-COORDINATION2018}}\footnote{\label{footnote-clarification-es}\revA{
Our notion of an event structure does not use the conflict relation~\cite{DBLP:journals/tcs/NielsenPW81},
which is used to express non-determinism. Indeed, we only use partial ordering.
Though it could be called a pomset~\revB{\cite{DBLP:journals/ijpp/Pratt86}},
we use this terminology to conform to previous research~\cite{ABDV-COORDINATION2018,JLAMP2019}.
In our model, non-determinism is provided by the environment:
then, a single event structure is used to describe \emph{one} possible system execution,
and results referring to multiple system executions universally quantify over all the possible event structures.
}}.
Following the general approach in \cite{ABDV-COORDINATION2018},
we enrich the event structure with information about the devices where events occur.

\begin{defi}[Situated event structure] \label{def:event-structure}
	A \emph{situated event structure} (ES) is a triplet $\ap{\eventS,\neigh,\devof}$ where: 
	\begin{itemize}
		\item
		$\eventS$
		is a countable set of \textbf{\emph{events}};
		\item
		$\neigh \; \subseteq \eventS \times \eventS$ is a \textbf{\emph{messaging}} relation from a \textbf{\emph{sender}} event to a \textbf{\emph{receiver}} event (these are also called \textbf{\emph{neighbour} events});
		\item $\devof : \eventS \to \deviceS$, where $\deviceS$ is a finite set of device identifiers $\deviceId_i$, maps an event $\eventId \in \eventS$ to the device $\devof(\eventId) \in \deviceS$ where the event takes place.
	\end{itemize}
The elements of the triplet are such that:
	\begin{itemize}
		\item
		the transitive closure of $\neigh$ forms an irreflexive partial order $< \; \subseteq \eventS \times \eventS$,  called \textbf{\emph{causality}} relation
		(an event $\eventId$ is in the \textbf{\emph{past}} of another event $\eventId'$ if $\eventId < \eventId'$,
		in the \textbf{\emph{future}} if $\eventId' < \eventId$,
		or \textbf{\emph{concurrent}} otherwise);
		\item 		for any $\deviceId \in \deviceS$, the projection of the ES to the set of events $\eventS_\deviceId = \bp{\eventId \in \eventS ~\mid~ \devof(\eventId) = \deviceId}$ forms a well-order, i.e., a sequence $\eventId_0\neigh\eventId_1\neigh\eventId_2\neigh\dots$.
	\end{itemize}
\addrev{Additionally, we introduce the following notation:}
	\begin{itemize}
	\item $\receivers(\eventId) = \{ d(\eventId') \mid \eventId~\neigh~\eventId'\}$
	 to denote the set of receivers of $\eventId$, i.e., the devices receiving a message from $\eventId$;
	\item $\tpast{\eventId_0}=\{\eventId: \eventId<\eventId_0 \}$ to denote the past event cone of $\eventId_0$ (\revA{which is a finite set, since we assume the system has a starting point in time at any device});
	\item $\tfuture{\eventId_0}=\{\eventId: \eventId_0<\eventId\}$ to denote the future event cone of $\eventId_0$;
	\item $\trestrict{X}{E'}$ to denote the projection of a set, function, or ES $X$ to the set of events $E' \subseteq \eventS$. Note that the projection of an event structure to the future event cone of an event is still a well-formed ES.
	\end{itemize}
\end{defi}

\definecolor{viridisYellow}{RGB}{253,231,37}
\definecolor{viridisGreen}{RGB}{94,201,98}
\definecolor{viridisDarkgreen}{RGB}{33,145,140}
\definecolor{viridisBlue}{RGB}{59,82,139}
\definecolor{viridisPurple}{RGB}{68,1,84}

\newcommand{\colorFuture}{viridisBlue!30}
\newcommand{\colorPresent}{viridisYellow!60}
\newcommand{\colorPast}{viridisGreen!30}

\newcommand{\colorRegionA}{viridisBlue!40}
\newcommand{\colorRegionB}{viridisYellow!60}
\newcommand{\colorRegionC}{viridisGreen!40}
\newcommand{\colorRegionD}{viridisPurple!40}
\newcommand{\colorRegionE}{viridisDarkgreen!40}
\newcommand{\colorRegionOpacity}{0.6}

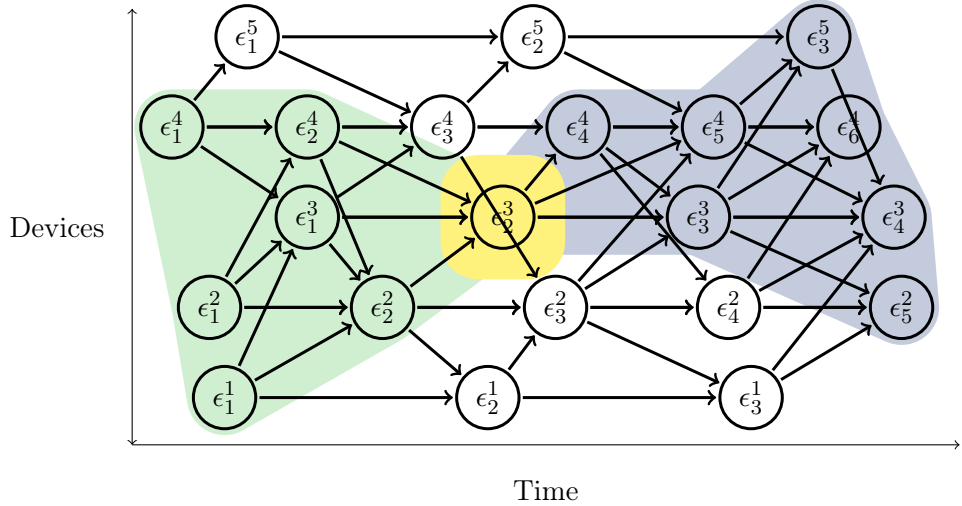
\begin{figure}[h]
\newcommand{\toAbove}{1.2cm}
\newcommand{\toRight}{1cm}
\centering
\adjustbox{max width=\esResizeFactor\textwidth}{
\begin{tikzpicture}
\node[] (origin) at (0,0) {};

\foreach \dev/\evs/\offx/\freq in {%
1/3/ -0.8 / 3.5,%
2/5/ -0.3 / 2.3,%
3/4/ 0.2 / 2.6,%
4/6/ -0.3 / 1.8,%
5/3/ -0.8 / 3.8}
  \foreach \ev in {1,...,\evs}
        \node[present] (d\dev\ev) [above right = \dev*\toAbove and \ev*\toRight*\freq,xshift=-1cm+\offx*2cm,yshift=-0.9cm] {$\eventId_\ev^\dev$}; 

\makeatletter
\long\def\ifnodedefined#1#2#3{%
    \@ifundefined{pgf@sh@ns@#1}{#3}{#2}%
}
\makeatother
\foreach \ref in {d3e2} {
\foreach \a in {1,...,5} {
 \foreach \ea in {1,...,10} {
  \foreach \b in {1,...,5} {
  \foreach \eb in {1,...,10} {
  \newtoggle{closer}
   \ifnodedefined{d\a\ea}{
   	  \ifnodedefined{d\b\eb}{
   	  	\ifthenelse{\NOT \equal{\a}{\b} \OR \NOT \equal{\ea}{\eb} }{ 
		  \draw[->,evlink]
     		let \p1=(d\a\ea), \p2=(d\b\eb),
       		  \n3={\y2-\y1},
     		  \n2={\x2-\x1},
     		  \n1={veclen(\x2-\x1,\y2-\y1)} in
     		{\ifdim \n1 < 3.2cm
     		 \ifdim \n2 > 0.8cm
     		 (d\a\ea) -- (d\b\eb) 
     		 \fi   		 
     		 \fi};   	  
     	}{}
   	  }{}
   }{}
  }
}
}
}
}

\foreach \a in {1,...,5} {
    \foreach \ea in {1,...,10} {
        \pgfmathtruncatemacro\next{\ea+1};
        \ifnodedefined{d\a\next}{
            \draw[->,evlink] (d\a\ea) -- (d\a\next);
        }
    }
}

\begin{scope}[on background layer]
    \filldraw[\colorFuture, line join=round, line width=1cm] plot coordinates{(d32) (d44) (d45) (d53) (d46) (d34) (d25) (d25) (d33) (d32)}--cycle;
    \filldraw[\colorPast, line join=round, line width=1cm] plot coordinates{(d32) (d42) (d41) (d21) (d11) (d22) (d32)}--cycle;
    \filldraw[\colorPresent, line join=round, line width=1cm] plot coordinates{(d32.north west) (d32.north east) (d32.south east) (d32.south west)}--cycle;
    
\end{scope}

\draw[<->] (0,0) -- node[xshift=-1cm]{Devices} (0,5.8);
\draw[->](0,0) -- node[yshift=-0.6cm]{Time} (11,0);
\end{tikzpicture}
}
\caption{Example of an event structure. 
In the node labels, superscripts denote device identifiers, while subscripts are progressive numbers denoting subsequent rounds at the same device. 
The blue (resp. green) background denotes the future (resp. past) of a reference event denoted with a yellow background.
}
\label{fig:event-struct}
\end{figure}

\noindent An example of ES is given in \Cref{fig:event-struct},
events denote computation rounds.
Notice that self-messages
(i.e., messages from an event to the next on the same device)
can be used to model persistence of \emph{state} over time.
\addrev{
Also, notice that two subsequent events at some device (e.g., $\eventId_4^4$ and $\eventId_5^4$ of device $4$ in \Cref{fig:event-struct})
may share a same sender event from a neighbour ($\eventId_2^3$ in this case):
in a real system, this could be due to two distinct communication acts with the same message,
as well as to a mechanism by which the receiving device reuses the most recently received message
from that neighbour through some \emph{message retention} policy
(which is an implementation mechanism useful to support stability of neighbourhoods).
}

\revA{
\begin{rem}[Communication and distributed execution]\label{remark:comm}
The description of a system execution as an event structure as per Definition~\ref{def:event-structure} abstracts from details regarding the actual communication and scheduling mechanisms used in a deployed system.
Concrete communication mechanisms may include point-to-point network channels, based on wired or wireless technologies, broadcasts, intermediaries (like, e.g., the cloud), or even stigmergic (i.e., environment-mediated) means.
The computation rounds may be scheduled 
 at a fixed frequency, 
 using a particular time distribution, or 
 reactively (e.g., in response to new sensor values or reception of messages from neighbours).
A more in-depth discussion of how aggregate computing systems may be deployed and executed can be found in~\cite{casadei2020pulverization,lmcs-timefluid}.
\end{rem}
}

In the computation model we consider, based on \cite{ABDV-COORDINATION2018},
each event $\eventId$ represents the execution of a program taking all incoming messages,
and producing an outgoing message
(sent to all neighbours)
and a result value associated with $\eventId$.
Such ``map'' of result values across all events defines a computational field, as follows.

\begin{defi}[Computational field] \label{def:field}
Let $\aEventS = \ap{\eventS,\neigh,\devof}$ be an event structure.
A \emph{computational field} on $\aEventS$ is a function $\field: \eventS \to \ValueSet$
 that maps every event $\eventId$ in $\eventS$ (also called the \emph{domain} of the field)
 to some value in a value set $\ValueSet$.
\end{defi}
Computational fields are essentially the ``distributed values'' which our model deals with;
hence computation is captured by the following definition.

\begin{defi}[Field computation]\label{def:field-computation}
Let $\aEventS = \ap{\eventS,\neigh,\devof}$ be an event structure, 
 and denote $\FieldSet_{\eventS,\ValueSet}$ the set of fields on domain $\eventS$ \revA{and co-domain} $\ValueSet$, i.e., $\FieldSet_{\eventS,\ValueSet} = \{ f: \eventS \to \ValueSet \}$.
\revA{Given two sets of values $\ValueSet,\ValueSet'$, a \emph{field computation}  over $\aEventS$
 is a function $\FieldComp_\aEventS: \FieldSet_{\eventS,\ValueSet} \to \FieldSet_{\eventS,\ValueSet'}$ mapping an input field to an output field on the same domain of $\aEventS$ (but possibly on a different co-domain).}
\end{defi}
This definition naturally extends to the case of zero or multiple input fields.

\addrev{
Now, we define the notion of a \emph{field-based program}, which we denote as a construct expressing a computation on any possible \emph{environment}, where an environment can be modelled by an event structure and fields over it denoting environmental values perceivable by devices
 (e.g., temperature fields would assign a temperature value to each event).
}

\begin{defi}[Field-based operator]\label{def:field-program}
A \emph{field-based operator} (or \emph{field-based program}) is a function $\Program$ taking an event structure as input and yielding the field computation that would occur on it, namely $\Program(\aEventS)=\FieldComp_\aEventS$.
\end{defi}

\revA{
In other words, a field-based operator works as a ``program'' (design-time): it applies to a certain event structure to generate a resulting field computation (run-time).
}
\labeltext{on-program-vs-operator}{\revB{
It would also be correct to say that 
 a field-based program
 provides an implementation of a field-based operator, similarly to 
 how, e.g., \emph{quick-sort} 
 provides an implementation of a sorting operator on lists.%
}}%
\labeltext{on-operator-computability}{\revB{
Notice also that we restrict our analysis to \emph{computable} field-based operators as per previous work~\cite{ABDV-COORDINATION2018}.
}}

\addrev{
There exist core languages and full-fledged programming languages for conveniently expressing field-based programs: these are known as \emph{field calculi} and \emph{aggregate programming languages}, respectively~\cite{JLAMP2019}. One of these is used in \Cref{ssec:ac-impl} to implement our adaptive spatial sampling algorithm. 
However, as the following example demonstrates,
 field-based programs can also be expressed, in our framework of event structures, by a \emph{local} perspective, in terms of how inputs and ingoing messages at one event are mapped to an output and outgoing messages.

\begin{exa}[Gradient field computation and operator\label{example:gradient}]
Term \emph{gradient} commonly refers to a kind of distributed data structures for estimating the distance from any device in a network to its closest \emph{source} device,
 and a family of distributed algorithms for building them~\cite{audrito2017ult}, which are very useful for implementing self-organising systems~\cite{Nagpal2002programmable-self-assembly,TOMACS2018,DBLP:conf/saso/WolfH07}.

In our framework, a distributed algorithm for building gradients (\emph{gradient operator}) can be modelled as 
 a field-based operator
 $\Program_G$
 mapping any event structure $\aEventS = \ap{\eventS,\neigh,\devof}$ 
 to a gradient field computation $\FieldComp_G$ over it.
A \emph{gradient field computation} $\FieldComp_G: \FieldSet_{\aEventS,\BoolType \times \MetricType} \to \FieldSet_{\aEventS,\R}$ 
 is essentially a map:
\begin{itemize}
\item from an input field $f_i: \FieldSet_{\aEventS,\BoolType \times \MetricType}$ mapping each event $\eventId \in \eventS$ with a pair $(\mathit{source},\mathit{metric}) \in \BoolType \times \MetricType$,
where $\MetricType \subseteq \eventS \times \R$,
 denoting \emph{source} events/devices (those where $\mathit{source}=\True$), distinguished
 from non-source events/devices (those where $\mathit{source}=\False$),
 and a metric associating neighbouring events to an estimation of the corresponding spatial distance, $\mathit{metric}(\eventId) = \{ 
 (\eventId',x) \mid 
 \eventId' \neigh \epsilon \land x \in \R \}$; 
\item 
 to the output field stabilising (cf. \Cref{ssec:self-stab}) to the minimum distances to source devices.
\end{itemize}

A simple operator for a gradient computation could be implemented through the following function, local to an event $\eventId_0 \in \eventS$ receiving from a possibly empty set of sender events $\eventId_i$ (with $\eventId_0$ denoting the sender event at the same device, if any, and with $\eventId_{i>0}$ the sender events from other devices), of the input field's value and the message set $\MsgSet=\{ \eventId_i \mapsto g_i \}$ providing the neighbours' current gradient estimates:
$$
g(\mathit{source},\{ \eventId_i \mapsto \mathit{dist}_i \}, \{ \eventId_i \mapsto g_i \}) = 
 \begin{cases}
 0 & \mathit{source}=\True \\
 \min_{i > 0} \{ g_i + dist_i \} & \mathit{source}=\False \land  \{ \eventId_{i>0} \} \neq \emptyset \\
 +\infty & \text{otherwise}
 \end{cases}
$$
An example of the induced computation is shown in \Cref{fig:stabilising-field}, assuming a simple metric where $\mathit{dist}_0=0$ and $\mathit{dist}_{i>0} = 1$.
\end{exa}
}

\subsection{Self-stabilisation}
\label{ssec:self-stab}

We now provide the definitions necessary to model \emph{self-stabilisation} following the approach in \cite{TOMACS2018}.
Namely, 
the following definitions capture the idea of \emph{adaptiveness} whereby as the environment of computation stabilises,
then the result of computation stabilises too,
and such a result does not depend on previous transitory changes.

\begin{defi}[Static environment]\label{def:static-env}
An event structure $\aEventS = \ap{\eventS,\neigh,\devof}$ 
 is said to be a \emph{static environment} if it has stable topology, namely all events of a given device
 always share the same set of receivers, i.e., $\forall \eventId, \eventId', \devof(\eventId)=\devof(\eventId') \Rightarrow 
 \receivers(\eventId)=\receivers(\eventId')$.
\end{defi}

\revA{
Note that, following the approach in~\cite{TOMACS2018}, we introduce the notion of a static environment to capture the eventual situation
 in which the environment stops perturbing the system.
This is instrumental to rely on an abstract characterisation of self-stabilising computations,
 which are those in which the system keeps intercepting changes in the environment and adapting to them: whenever (and if) the environment becomes static, one can observe the result of that adaptation that eventually establishes.
}

\begin{defi}[Stabilising environment]\label{def:stab-env}
An event structure $\aEventS = \ap{\eventS,\neigh,\devof}$ 
 is said to be a \emph{stabilising environment}
 if it 
 is eventually static,
 i.e.,
 $\exists \eventId_0 \in \eventS 
 \textrm{~such that~} 
 \tfutureof{\aEventS}{\eventId_0} 
 = \ap{\tfutureof{\eventS}{\eventId_0},
  \tfutureof{\neigh}{\eventId_0}, 
  \tfutureof{\devof}{\eventId_0} 
}$ is static. In this case we say it is \emph{static since} event $\eventId_0$.
\end{defi}

\begin{defi}[Stabilising field]\label{def:selfstab-field}
Let event structure $\aEventS = \ap{\eventS,\neigh,\devof}$ be a stabilising environment, static since event $\eventId_0$.
A field $\field: \eventS \to \ValueSet$
 is said \emph{stabilising} 
 if it eventually provides stable output \revA{(an output that does not change since some round)}, i.e.,
 $\exists \eventId > \eventId_0 \textrm{~such that~} \forall \eventId' > \eventId, \eventId'' > \eventId \textrm{~it holds that~} d(\eventId')=d(\eventId'') \implies f(\eventId')=f(\eventId'')$.
\end{defi}

\addrev{
An example of a stabilising field, which can also be thought as being generated by a gradient computation (cf. \autoref{example:gradient}), is provided in \Cref{fig:stabilising-field}.
The environment is static since event $\eventId_1^3$ (every event in the future event cone of $\eventId_1^3$ has the same set of receivers), and from event $\eventId_2^3$ (excluded) it holds that each device does not change the value it produces in its rounds.
}

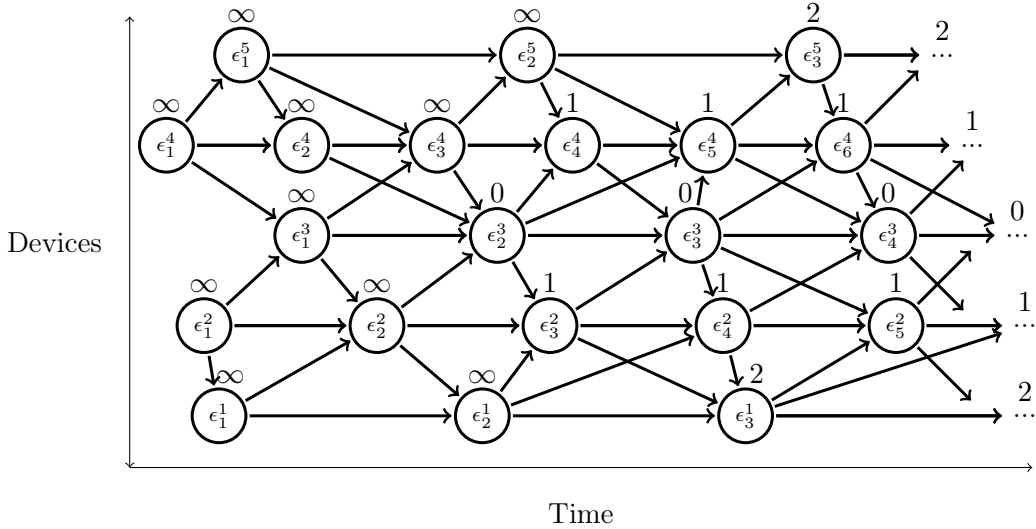
\begin{figure}[h]
\newcommand{\toAbove}{1.2cm}
\newcommand{\toRight}{1cm}
\centering
\adjustbox{max width=\esResizeFactor\textwidth}{
\begin{tikzpicture}
\node[] (origin) at (0,0) {};

\foreach \dev/\evs/\offx/\freq in {%
1/3/ -0.8 / 3.5,%
2/5/ -0.3 / 2.3,%
3/4/ 0.2 / 2.6,%
4/6/ -0.3 / 1.8,%
5/3/ -0.8 / 3.8}
  \foreach \ev in {1,...,\evs}
        \node[present] (d\dev\ev) [above right = \dev*\toAbove and \ev*\toRight*\freq,xshift=-1cm+\offx*2cm,yshift=-0.8cm] {\scriptsize$\eventId_\ev^\dev$}; 
        
\makeatletter
\long\def\ifnodedefined#1#2#3{%
    \@ifundefined{pgf@sh@ns@#1}{#3}{#2}%
}
\makeatother

\foreach \ref in {d3e2} {
\foreach \a in {1,...,5} {
 \foreach \ea in {1,...,10} {
  \foreach \b in {1,...,5} {
  \foreach \eb in {1,...,10} {
  \newtoggle{closer}
   \ifnodedefined{d\a\ea}{
   	  \ifnodedefined{d\b\eb}{
   	  	\ifthenelse{\NOT \equal{\a}{\b} \OR \NOT \equal{\ea}{\eb} }{ 
		  \draw[->,evlink]
     		let \p1=(d\a\ea), \p2=(d\b\eb),
       		  \n3={abs(\y2-\y1)},
     		  \n2={\x2-\x1},
     		  \n1={veclen(\x2-\x1,\y2-\y1)} in
     		{\ifdim \n1 < 3.2cm
     		 \ifdim \n2 > 0.1cm
     		 \ifdim \n3 < 2.0cm
     		 (d\a\ea) -- (d\b\eb) 
     		 \fi   		 
     		 \fi
     		 \fi};   	  
     	}{}
   	  }{}
   }{}
  }
}
}
}
}

\draw[->,evlink] (d12) -- (d24);
\draw[->,evlink] node (d14) [right=3cm of d13] {...} (d13) -- (d14);
\draw[->,evlink] node (d26) [right=1cm of d25] {...} (d25) -- (d26);
\draw[->,evlink] node (d35) [right=1cm of d34] {...} (d34) -- (d35);
\draw[->,evlink] node (d47) [right=1cm of d46] {...} (d46) -- (d47);
\draw[->,evlink] node (d54) [right=1cm of d53] {...} (d53) -- (d54);

\draw[->,evlink] node (d26b) [below right=1cm of d25] {} (d25) -- (d26b);
\draw[->,evlink] node (d26c) [above right=1cm of d25] {} (d25) -- (d26c);
\draw[->,evlink] node (d35b) [below right=1cm of d34] {} (d34) -- (d35b);
\draw[->,evlink] node (d35c) [above right=1cm of d34] {} (d34) -- (d35c);
\draw[->,evlink] node (d47b) [above right=1cm of d46] {} (d46) -- (d47b);

\draw[->,evlink] (d46) -- (d35);
\draw[->,evlink] (d13) -- (d26);

\foreach \a in {1,...,5} {
 \foreach \ea in {1,...,10} {
 \pgfmathtruncatemacro\next{\ea+1};
 \ifnodedefined{d\a\next}{
 	\draw[->,evlink] (d\a\ea) -- (d\a\next);
 }
}
}

\draw[<->] (0,0) -- node[xshift=-1cm]{Devices} (0,6);
\draw[->](0,0) -- node[yshift=-0.6cm]{Time} (12,0);

\node[] () [above=-0.1cm of d11, xshift=0.15cm] {$\infty$};
\node[] () [above=-0.1cm of d12] {$\infty$};
\node[] () [above=-0.1cm of d13, xshift=0.15cm] {$2$};
\node[] () [above=-0.1cm of d14] {$2$};

\node[] () [above=-0.1cm of d21] {$\infty$};
\node[] () [above=-0.1cm of d22] {$\infty$};
\node[] () [above=-0.1cm of d23] {$1$};
\node[] () [above=-0.1cm of d24] {$1$};
\node[] () [above=-0.1cm of d25] {$1$};
\node[] () [above=-0.1cm of d26] {$1$};

\node[] () [above=-0.1cm of d31] {$\infty$};
\node[] () [above=-0.1cm of d32] {$0$};
\node[] () [above=-0.1cm of d33, xshift=-0.1cm] {$0$};
\node[] () [above=-0.1cm of d34] {$0$};
\node[] () [above=-0.1cm of d35] {$0$};

\node[] () [above=-0.1cm of d41] {$\infty$};
\node[] () [above=-0.1cm of d42] {$\infty$};
\node[] () [above=-0.1cm of d43] {$\infty$};
\node[] () [above=-0.1cm of d44] {$1$};
\node[] () [above=-0.1cm of d45] {$1$};
\node[] () [above=-0.1cm of d46] {$1$};
\node[] () [above=-0.1cm of d47] {$1$};

\node[] () [above=-0.1cm of d51] {$\infty$};
\node[] () [above=-0.1cm of d52] {$\infty$};
\node[] () [above=-0.1cm of d53] {$2$};
\node[] () [above=-0.1cm of d54] {$2$};

\end{tikzpicture}
}
\caption{\addrev{Example of a stabilising field. We use labels above the nodes to denote the values computed in the corresponding events, assuming the program is a gradient operator as per~\autoref{example:gradient}.}
}
\label{fig:stabilising-field}
\end{figure}

\begin{defi}[Stabilising computation]\label{def:selfstab-computation}
A field computation $\FieldComp_\aEventS = \FieldSet_{\eventS,\ValueSet} \to \FieldSet_{\eventS,\ValueSet}$ is said \emph{stabilising} if,
 when applied to a stabilising input field,
 it yields a stabilising output field.
\end{defi}

\begin{defi}[Self-stabilising operator]\label{def:selfstab-operator}
A field-based operator (or program) $\Program$ is said self-stabilising, if in any stabilising environment $\aEventS$ it yields a stabilising computation $\FieldComp_\aEventS$ such that, for any pair of input fields $f_1,f_2$ eventually equal, i.e. $\trestrict{f_1}{\tfuture{\eventId}}=\trestrict{f_2}{\tfuture{\eventId}}$ for some event $\eventId$, their output is eventually equal too, i.e., there exists a $\eventId'>\eventId$ such that $\trestrict{\FieldComp_\aEventS(f_1)}{\tfuture{\eventId'}}=\trestrict{\FieldComp_\aEventS(f_2)}{\tfuture{\eventId'}}$
\end{defi}

\labeltext{on-decidability}{\revB{%
Notice that universally quantifying over event structures, i.e., considering infinitely many system executions,
makes finding decision procedures for properties like stabilisation undecidable \emph{in general}.
However, this does not prevent us from reasoning about such properties for a specific program, as we carry on in this paper---and as developed in previous works, e.g., in \cite{TOMACS2018}.
\footnote{On the other hand, note that most of our definitions could be given considering finite runs, where proving decidability could be easier---but this is not developed for the sake of generality.}
}}

\subsection{Problem Definition}
\label{ssec:problem-def}

We start by introducing the notion of \emph{regional partition},
which is a finite set of non-overlapping contiguous clusters of devices:
a notion that prepares the ground to that of an \emph{\sampling{}} which we introduce in this paper.

\begin{defi}[Regional partition field\revA{, contiguous regions}]\label{def:es-partition}
Let $\aEventS = \ap{\eventS,\neigh,\devof}$ 
 be a stabilising environment
 static since event $\eventId_0$.
A \emph{regional partition field} is a stabilising field $\field: \eventS \to \ValueSet$ on $\aEventS$ such that:
 \begin{itemize}
 \item \textbf{(finiteness)} the image $\Img(f)=\{f(x) \mid x \in \eventS\}$ is a finite set of values; 
 \item \textbf{(eventual contiguity)} there exists an event $\eventId'_0 > \eventId_0$ such that for any pair of events $\eventId_1, \eventId_n \in \tfuture{\eventId'_0}$, \revA{$f(\eventId_1)=f(\eventId_n)$ implies that there is a sequence of events $\eventId_1 < ... < \eventId_n$ connecting $\eventId_1$ to $\eventId_n$ where $f(\eventId_i)=f(\eventId_1)=f(\eventId_n) \; \forall 1 \leq i \leq n$}.
 \end{itemize}
Note that the set of domains of regions induced by $f$ is defined by $\Regions(f)=\{f^{-1}(v) : v\in\Img(f)\}$.
\revA{Moreover, given two regions $\eventS, \eventS' \in \Regions(f)$, we say that they are \emph{contiguous} if $\exists \eventId \in \eventS, \eventId' \in \eventS' : \eventId \neigh \eventId' \lor \eventId' \neigh \eventId$.}
\end{defi}
An example of a regional partition field is shown in \Cref{fig:partition-field}.
Notice that for any pair of events in the same space-time region there exists a path of events entirely contained in that region.
Also, notice that, by this definition, different disjoint regions denoted by the same value $r$ are not possible.

\begin{defi}[\Sampling]\label{def:sampler}
An \emph{\sampling} is a stabilising computation $\FieldComp_S: \FieldSet_{\eventS,\ValueSet} \to \FieldSet_{\eventS,\ValueSet}$
 that, given an input field to be sampled,
 yields as output a regional partition field. 
\end{defi}

\begin{figure}[h]
\newcommand{\toAbove}{1.2cm}
\newcommand{\toRight}{1cm}
\centering
\adjustbox{max width=\esResizeFactor\textwidth}{
\begin{tikzpicture}
\node[] (origin) at (0,0) {};

\foreach \dev/\evs/\offx/\freq in {%
1/3/ -0.8 / 3.5,%
2/5/ -0.3 / 2.3,%
3/4/ 0.2 / 2.6,%
4/6/ -0.3 / 1.8,%
5/3/ -0.8 / 3.8}
  \foreach \ev in {1,...,\evs}
        \node[present] (d\dev\ev) [above right = \dev*\toAbove and \ev*\toRight*\freq,xshift=-1cm+\offx*2cm,yshift=-0.8cm] {$\eventId_\ev^\dev$}; 

\makeatletter
\long\def\ifnodedefined#1#2#3{%
    \@ifundefined{pgf@sh@ns@#1}{#3}{#2}%
}
\makeatother

\foreach \ref in {d3e2} {
\foreach \a in {1,...,5} {
 \foreach \ea in {1,...,10} {
  \foreach \b in {1,...,5} {
  \foreach \eb in {1,...,10} {
  \newtoggle{closer}
   \ifnodedefined{d\a\ea}{
   	  \ifnodedefined{d\b\eb}{
   	  	\ifthenelse{\NOT \equal{\a}{\b} \OR \NOT \equal{\ea}{\eb} }{ 
		  \draw[->,evlink]
     		let \p1=(d\a\ea), \p2=(d\b\eb),
       		  \n3={abs(\y2-\y1)},
     		  \n2={\x2-\x1},
     		  \n1={veclen(\x2-\x1,\y2-\y1)} in
     		{\ifdim \n1 < 3.2cm
     		 \ifdim \n2 > 0.1cm
     		 \ifdim \n3 < 2.0cm
     		 (d\a\ea) -- (d\b\eb) 
     		 \fi   		 
     		 \fi
     		 \fi};   	  
     	}{}
   	  }{}
   }{}
  }
}
}
}
}

\foreach \a in {1,...,5} {
 \foreach \ea in {1,...,10} {
 \pgfmathtruncatemacro\next{\ea+1};
 \ifnodedefined{d\a\next}{
 	\draw[->,evlink] (d\a\ea) -- (d\a\next);
 }
}
}

\draw[->] node (d54b) [right=1cm of d53] {...} ;
\draw[->] node (d26c) [right=1cm of d25] {...} ;
\draw[->] node (d35b) [right=1cm of d34] {...} ;
\draw[->] node (d14b) [right=2.8cm of d13] {...} ;
\draw[->] node (d47b) [right=1cm of d46] {...} ;

\begin{scope}[on background layer]
    \filldraw[draw=\colorRegionA,fill=\colorRegionA, line join=round, line width=1cm,fill opacity=\colorRegionOpacity,draw opacity=\colorRegionOpacity] plot coordinates{
	(d11) (d22) (d26) (d47) (d44) (d11)
}--cycle;
    \filldraw[draw=\colorRegionC,fill=\colorRegionC, line join=round, line width=1cm,fill opacity=\colorRegionOpacity,draw opacity=\colorRegionOpacity] plot coordinates{
    (d42) (d52) (d54) (d52.south) (d43) (d42)
}--cycle;
    \filldraw[draw=\colorRegionB,fill=\colorRegionB, line join=round, line width=1cm,fill opacity=\colorRegionOpacity,draw opacity=\colorRegionOpacity] plot coordinates{
    (d13.south) (d26b.south) (d26b.north) (d13.north)
}--cycle;
    \filldraw[draw=\colorRegionA,fill=\colorRegionA, line join=round, line width=1cm,fill opacity=\colorRegionOpacity,draw opacity=\colorRegionOpacity] plot coordinates{
	(d51.south west) (d51.south east) (d51.north east) (d51.north west)
}--cycle;
\end{scope}

\draw[<->] (0,0) -- node[xshift=-1cm]{Devices} (0,6);
\draw[->](0,0) -- node[yshift=-0.6cm]{Time} (12,0);

\draw[->,evlink] (d12) -- (d24);
\draw[->,evlink] (d13) -- (d26c);
\end{tikzpicture}
}
\caption{Example of a regional partition field with regions $r_{blue}$,  $r_{green}$, $r_{yellow}$, $r_{white}$ (the background is used to denote the output of the field).
Notice that contiguity does not hold everywhere and anytime but only since event $\eventId_2^3$. 
}
\label{fig:partition-field}
\end{figure}
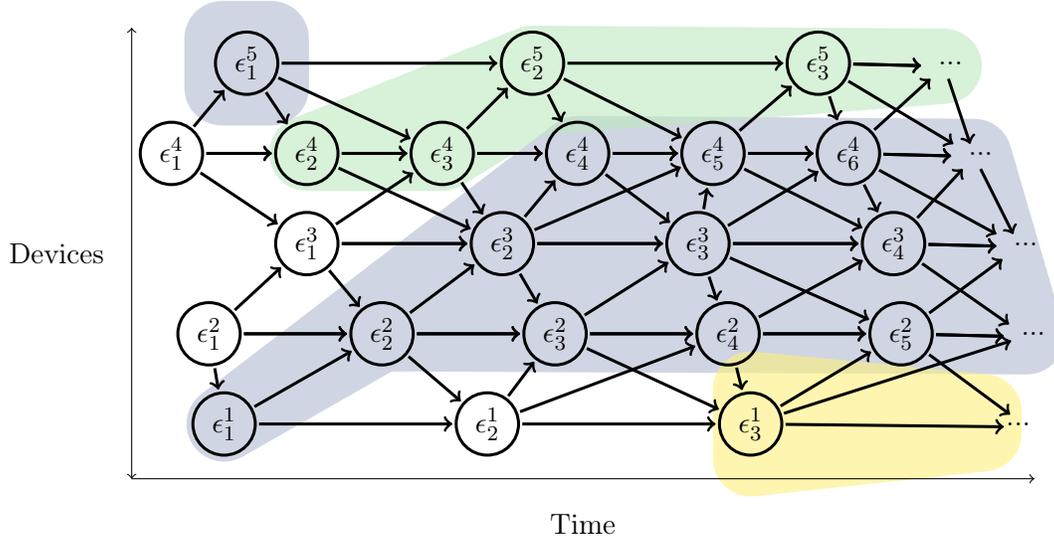

\noindent Once we have defined an \sampling{} process 
 in terms of its inputs, outputs, and stabilising dynamics,
 we need a way to measure the error introduced by the \sampling{}.
To this purpose, we introduce the notion of a \emph{stable snapshot},
 namely a field 
 consisting of a sample of one event per device
 from the stable portion of a stabilising field. 

\begin{defi}[Stable snapshot] \label{def:spatialSection}
Let $\aEventS = \ap{\eventS,\neigh,\devof}$ be an event structure, 
and $\field: \eventS \to \ValueSet$ be a stabilising field on $\aEventS$ which provides stable output from $\eventId_0 \in \eventS$.
We define a \emph{stable snapshot} of field $\field$
 as a field obtained by restricting $\field$ to a subset of events in the future event cone of $\eventId_0$ \emph{and}
 with exactly one event per device, i.e., a field $\field_S: \eventS_S \to \ValueSet$ such that $\eventS_S \subseteq \tfuture{\eventId_0}$,
 \emph{and} 
 $\forall \eventId, \eventId' \in \eventS_S: d(\eventId) = d(\eventId') \implies \eventId=\eventId'$,
 \emph{and} $\forall \eventId \in \tfuture{\eventId_0}, \exists \eventId' \in \eventS_S: d(\eventId') = d(\eventId)$.
\end{defi}

\newcommand{\ssed}[0]{error-distance}

\begin{defi}[Stable snapshot \ssed{}] \label{def:stablesnapshotdist}
\revA{We call \emph{stable snapshot \ssed} any metric $\fmetric: \FieldSet_{\eventS,\ValueSet} \times \FieldSet_{\eventS,\ValueSet} \to \mathbb{R}^+_0$
 over stable snapshots that feature same domain (event structure) and codomain (set of values)}.
\end{defi}

\revT{We are now able to characterise adequacy properties for a sampling operator, intuitively capturing the fact that sampling correctly trades-off the size of regions with their accuracy.
We first start by introducing a notion that handles accuracy, stating that any of the produced regions won't cause the \ssed{} to be over a certain threshold.}

\begin{defi}[\Sampling{} error]\label{def:error}
Let $\FieldComp_\aEventS: \FieldSet_{\eventS,\ValueSet} \to \FieldSet_{\eventS,\ValueSet}$ be an \sampling{}, and consider an input field $f_i: \eventS \to \ValueSet$ and corresponding output regional partition $f_o: \eventS \to \ValueSet$.
We say that \emph{$f_o$ samples $f_i$ within error $\samplingErr$ according to \ssed{} $\fmetric$}, if the \ssed{} of stable snapshots of $f_i$ and $f_o$ in any region is not bigger than $\samplingErr$, that is:
let $f_i^s$ and $f_o^s$ be stable snapshots of $f_i$ and $f_o$, then for any region $\eventS'\in\Regions(f_o^s)$, we have $\fmetric(\trestrict{f_i^s}{\eventS'},\trestrict{f_o^s}{\eventS'})\leq\samplingErr$.
\end{defi}
\revT{
Note that \emph{accuracy} can be generally achieved simply by partitions defining many small regions---up to the corner case in which all regions include just one device, hence trivially induce zero \ssed{}.
Therefore, we are also interested in \emph{efficiency}, namely the ability of a regional partition to rely on as few regions as possible.
Without a centralised approach, however, partitioning is necessarily sub-optimal,
since it can rely only on local interaction/competition among regions,
hence it should be expected that some regions will stop ``expanding'' as they reach a smaller threshold.
Additionally, there can also be corner cases where regions with very small \ssed{} are created, e.g., because what remains to be covered in an iterative selection of regions is simply a very small part of the network, or one with rather uniform values introducing little sampling errors.
What we may require from an adequate sampling operator, however, is that such regions are somewhat not the norm.
This is formally captured by the following definition, essentially introducing a ``lower bound'' for the \ssed{} of regions.}
 
\begin{defi}[Local optimality of a regional partition]\label{def:optimal}
Let $\FieldComp_\aEventS: \FieldSet_{\eventS,\ValueSet} \to \FieldSet_{\eventS,\ValueSet}$ be an \sampling{}, consider an input field $f_i$ and corresponding output regional partition $f_o$ such that 
$f_o$ samples $f_i$ within error $\samplingErr$ according to distance $\fmetric$\revA{, and
denote with $f_i^s$ and $f_o^s$ the stable snapshots of $f_i$ and $f_o$, respectively (cf. Definition \ref{def:error})}.
\revT{We say that $f_o$ is \emph{locally optimal} under error $\samplingErr$ and with efficiency $k$ ($k>0$) if all pairs of contiguous regions $\eventS',\eventS''\in\Regions(f_o)$ are such that
$\fmetric(\trestrict{f_i^s}{\eventS'\cup\eventS''},\trestrict{f_o^s}{\eventS'\cup\eventS''})\geq k \cdot \samplingErr{}$.}
\end{defi}
\revT{For example, we will show that the algorithm we propose guarantees $k = 0.5$ (see \Cref{ssec:formal-results}).
Note that we call this notion ``local optimality'' to stress the fact that an identified partition is not necessarily the best one that could be found, but it is one that cannot be significantly improved with a small change, such as combining two regions---a small improvement is possible, depending on the efficiency factor $k$.
This notion well fits our goal of dealing with dynamic phenomena and large-scale environments,
 where one is more geared towards finding good heuristics for self-organising behaviour.}

So, we are now ready to define the goal operator for this paper.

\begin{defi}[Effective sampling operator] \label{def:optsampling}
An \emph{effective sampling operator} with \revT{efficiency $k$} is a self-stabilising operator $\Program_\samplingErr$, parametric in the error bound $\samplingErr$, such that in any stabilising environment $\aEventS$ and stabilising input $f_i$, a locally optimal regional partition \revT{with efficiency $k$} and within error $\samplingErr$ is produced.
\end{defi}

\newcommand{\algname}{{\textsc{AggregateSampler}}}
\newcommand{\pse}{path sampling error}

\section{Aggregate Computing-based Solution}\label{s:solution}

\addrev{
In this section, we define a \emph{space-based adaptive sampling} algorithm, called \algname{}, (\Cref{ssec:algo}),
discuss its implementation in \emph{aggregate computing}~\cite{BealIEEEComputer2015,JLAMP2019} (\Cref{ssec:ac-impl}),
and prove the algorithm is a self-stabilising, effective sampler with efficiency at least $k=0.5$ (\Cref{ssec:formal-results}).
}
\addrev{
The algorithm is defined in terms of the computational model described in \Cref{ssec:comp-model}, as well as its implementation, and is the one evaluated in \Cref{s:eval}.
The proofs are based on the definitions of \Cref{ssec:self-stab} and \Cref{ssec:problem-def}.
}

\subsection{\algname{} Algorithm for Adaptive Spatial Sampling}
\label{ssec:algo}

The problem of creating partitions in a self-organising way is very much related to a problem of multi-leader election
 ~\cite{FGCS2020-scr,pianini2022acsos-leader-election}.

\revT{Building on this idea, our approach starts by solving a \emph{sparse leaders election problem}~\cite{DBLP:books/mk/Lynch96}, for which self-stabilising solutions exist~\cite{Mo2020leader-election,DBLP:conf/podc/BurmanCCDNSX21,pianini2022acsos-leader-election}.
Leaders are used as \emph{samplers} of the input field. During the election of leaders/samplers, we associate them with larger and larger regions of ``follower devices'' that will provide the sampled value as output.
During execution of the algorithm, such regions will expand until the desired \ssed{} can be kept under the threshold $\eta$.
This process is managed so that there won't be any overlap with other regions, and so that no devices of the network remain outside of some region
(i.e., each device will follow exactly one leader).

To ensure that regions are connected, and won't overcome the threshold independently of the chosen leader, we adopt as \ssed{} one based on ``distance among devices'', as follows.
The algorithm can be configured to adopt any strategy that is able to turn input and output fields into a metric $m$ for devices: such a metric is as usual a function mapping a pair of neighbour devices to a non-negative real number, called the ``local sampling distance'' of the two devices---intuitively, the higher the physical distance of devices and the higher the difference of input values and output values of the two devices, the higher is $m$ for that pair.
Given this metric, any pair of devices of the network can be associated with a \emph{\pse}, which is the size of the shortest path (according to the metric) connecting the two devices.
The proposed algorithm will then produce regions adopting as \ssed{} $\fmetric$ the maximum \pse{} of any pair of devices in the region, and it will turn out that any pair of contiguous regions combined will necessarily give \ssed{} greater than $0.5*\eta$ (efficiency $k=0.5$).}

The algorithm is defined as follows (see \Cref{fig:algoexample-regions}):
\begin{enumerate}
	\item each device announces its candidature for leadership;
	\item each device propagates to its neighbours the candidature of the device it currently recognises as leader, its sampled value, and the \pse{} from it,
    fostering the expansion of its corresponding region;
    \item \label{algostep-discard-candidacy} devices discard candidatures whose \pse{} from the leader exceeds half the expected threshold (\revT{$\samplingErr/2$});
	\item \label{algostep-select-candidacy} in case multiple valid candidatures (i.e., those that are not discarded) reach a device, one is selected based on a \emph{competition policy}.
\end{enumerate}
\revA{The specific strategies for
    computing the local sampling distance
    and the leader competition policy are application-dependent---we will provide some instances in \Cref{s:eval}.
}

\begin{figure}[h]
\begin{tikzpicture}

\newcommand{\cRegion}[3]{
\begin{scope}[on background layer]
    \filldraw[draw=#1,fill=#1, line join=round, line width=1cm,fill opacity=\colorRegionOpacity,draw opacity=\colorRegionOpacity] plot coordinates{
	(n#2.north west) (n#3.north east)
	(n#3.south east) (n#2.south west)
}--cycle;
\end{scope}
}

\node[dev] (n0) [] {};
\foreach \dn [evaluate=\dn as \prev using \dn-1] in {1,...,10}
  \node[dev] (n\dn) [right=of n\prev] {};

\foreach \dn 
  in {0,...,9}
    \pgfmathsetmacro\succ{\dn + 1}
	\draw[<->] (n\dn) -- ($(n\succ)+(-0.3,0)$);

\cRegion{\colorRegionA}{0}{4}
\cRegion{\colorRegionB}{5}{9}
\cRegion{\colorRegionC}{10}{10}

\foreach \dn in {2,7,10}
  \node[dev,fill=black] (over\dn) at (n\dn) {};

\draw[<->,red] ($(n0)+(0,0.3)$) -- ($(n2)+(0,0.3)$) node [midway,yshift=0.3cm] {$\eta/2$};
\draw[<->,red] ($(n2)+(0,0.3)$) -- ($(n4)+(0,0.3)$) node [midway,yshift=0.3cm] {$\eta/2$};
\draw[<->,red] ($(n5)+(0,0.3)$) -- ($(n7)+(0,0.3)$) node [midway,yshift=0.3cm] {$\eta/2$};
\draw[<->,red] ($(n7)+(0,0.3)$) -- ($(n9)+(0,0.3)$) node [midway,yshift=0.3cm] {$\eta/2$};

\draw[<->,red] ($(n0)+(0,-0.3)$) -- ($(n4)+(0,-0.3)$) node [midway,yshift=-0.3cm] {$\eta$};
\draw[<->,red] ($(n5)+(0,-0.3)$) -- ($(n9)+(0,-0.3)$) node [midway,yshift=-0.3cm] {$\eta$};

\draw[<->,red] ($(n10)+(-0.3,-0.3)$) -- ($(n10)+(+0.3,-0.3)$) node [midway,yshift=-0.3cm] {$<< \eta$};

\end{tikzpicture}

\caption{\revA{
Example of a regional partitioning (with three contiguous regions) created by the algorithm on a simplified system where devices are arranged on a line. Notation: black dots denote the leaders/samplers; the coloured areas denote regions; and the red extension lines are used to denote the error-distances. Note that no device in a region can have path sampling error greater than $\eta/2$ with respect to the leader, and that very small regions can still exist in corner cases (e.g., the green region on the right).}
}
\label{fig:algoexample-regions}
\end{figure}
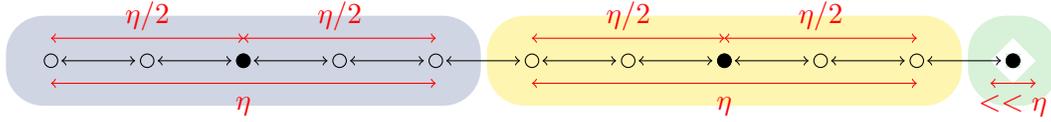

\subsubsection*{Competition and leader strength.}

Although competition among leaders could be realised in several ways,
many techniques may lead to non-self-stabilising behaviour:
for instance,
if the winning leader is selected randomly in the set of those whose error is under threshold,
regions may keep changing even in a static environment.
In this work, we propose a simple strategy:
every leader associates its candidature with the local value of a field that we call \emph{leader strength};
in case of competing candidatures,
the highest such value is selected as winner,
breaking the symmetry.
The leader strength can be of any orderable type,
and its choice impacts the overall selection of the regions by imposing a selection priority over leaders
(hence on region-generation points).
\addrev{
If two candidate leaders have the same strength,
 then we prefer the closest one.
If we are in the (unlikely) situation
 of perfect symmetry,
 with two equally-strong candidate leaders at the same distance,
 then their device identifier is used to break symmetry.
}

\subsubsection*{Region expansion and \pse{}.}

Inspired by previous work on distributed systems whose computation is independent of device distribution~\cite{TAAS2017},
the proposed approach essentially accumulates the \pse{} along the path from the leader device towards other devices
along a gradient (cf. \autoref{example:gradient}),
a distributed data structure that can be generated through self-stabilising computations~\cite{TOMACS2018}\revA{~(cf. \Cref{ssec:self-stab})}. 
We thus have two major drivers:
\begin{enumerate}
	\item the leader strength affects the creation of regions by influencing the positions of their source points;
	\item the \pse{} influences the expansion in space of the region across all directions,
	mandating its size and (along with the interaction with other regions) its shape.
\end{enumerate}
For instance,
a metric could be the absolute value of the difference in the perceived signal \revA{(e.g., a value sampled from a sensor---cf. \Cref{ssec:comp-model})} between
two devices:
devices perceiving very different values would tend not to cluster together (even if spatially close),
as they would perceive each other as farther away (leading to irregular shapes).

As the simulations in the next section verify, 
connecting region expansion with the error-distance
 (i.e., using the error-distance as a distance metric for gradient computation)
enables the determination of locally optimal sampling regions.
We recall that the local optimality property 
 means that all regions are essentially needed except for corner cases.

\subsection{Aggregate Computing-based Implementation}
\label{ssec:ac-impl}

\newcommand{\fmpName}[0]{expansionLogic}
\begin{figure}[th]
\begin{lstlisting}
// Definition of the record (product type) Sample + accessor functions %\label{code:protelis:candidacy:init}%
def Sample(symmetryBreaker, distance, leaderId) = [symmetryBreaker, distance, leaderId]
def breakSymmetry(sample) = sample.get(0)
def sampleDistance(sample) = sample.get(1)
def areaCenter(sample) = sample.get(2)
def discard() = Sample(POSITIVE_INFINITY, POSITIVE_INFINITY, POSITIVE_INFINITY) %\label{code:protelis:candidacy:end}%
// Logic to control the propagation of candidacies
def %\fmpName{}%(sample, localId, radius) = %\label{code:protelis:validate:def:init}%
  mux (areaCenter(sample) == localId || sampleDistance(sample) >= radius) {
    discard()
  } else {
    sample
  }%\label{code:protelis:validate:def:end}%

def %\algname{}%(mid, radius, symmetryBreaker, metric) { %\label{code:protelis:fun:init}%
  let local = Sample(-symmetryBreaker, 0, mid)
  areaCenter(
    share (received <- local) { %\label{code:protelis:share:init}%
      let candidacies = received.set(1, received.get(1) + metric()) %\label{code:protelis:distance}%
      let filtered = %\fmpName{}%(candidacies, mid, radius) %\label{code:protelis:validate:call}%
      min(local, foldMin(discard(), filtered)) %\label{code:protelis:min}%
    } %\label{code:protelis:share:end}%
  )
} %\label{code:protelis:fun:end}%
\end{lstlisting}
\caption{Source code of the algorithm.}
\label{fig:algorithm}
\end{figure}

\addrev{
An implementation of the algorithm expressed in the Protelis aggregate programming language~\cite{PianiniSAC2015} is shown in \Cref{fig:algorithm}.
In aggregate computing, a so-called \emph{aggregate program} such as the one shown in \Cref{fig:algorithm}, is repeatedly run by all the devices: it expresses a logic for mapping the \emph{local context} (given by sensor readings and messages from neighbours) to an \emph{output value} and an \emph{output message} to be sent to all the neighbours.
In other words, aggregate computing leverages the computational model described in \Cref{ssec:comp-model},
 where each event 
 denotes a full execution of the aggregate program
 against the event's inputs,
 determining the message payload passed to receiving events.

The core of the program is function \algname{},
 which consists of the following main elements:
\begin{itemize}
\item \emph{sampler candidacies are modelled as ordered triplets},
        using 3-element tuples, with corresponding accessor functions
        (\Cref{fig:algorithm}, Lines~\ref{code:protelis:candidacy:init}--\ref{code:protelis:candidacy:end}), of the elements:
        \begin{enumerate}
        \item \lstinline|simmetryBreaker|: a value used to break symmetry, capturing the ``strength'' of a candidacy;
        \item \lstinline|distance|: a value capturing the distance to the sampler node of a candidacy (e.g., computed through a gradient---cf. \autoref{example:gradient});
        \item \lstinline|leaderId|: holding the device identifier of the candidate sampler;
		\end{enumerate}         
\item function \texttt{\fmpName{}} (\Cref{fig:algorithm}, \ref{code:protelis:validate:def:init}--\ref{code:protelis:validate:def:end}) is defined to determine when a candidacy has to be discarded, i.e., when it comes from the device itself or when the distance of the candidate sampler is greater or equal than a \lstinline|radius| parameter;
\item function \texttt{\algname{}} (\Cref{fig:algorithm}, \ref{code:protelis:fun:init}--\ref{code:protelis:fun:end}) is the entry point of the algorithm, parametrised in terms of the executing device identifier (\lstinline|mid|), 
a maximum spatial range of candidacies (\lstinline|radius|),
a value to break symmetry (\lstinline|symmetryBreaker|),
and a \lstinline|metric| function providing distances to neighbours;
\item \lstinline|share(x <- init)\{ e \}| is a bidirectional communication construct~\cite{DBLP:journals/lmcs/AudritoBDPV20},
        that works as follows: the declared variable \lstinline|x|, which is set to \lstinline|init| at the first round,
        collects the evaluations of the overall \lstinline|share| expression in neighbour devices
        (including the device itself),
        and the new value for the current device
        (which is the data item that will be sent to neighbours)
        is obtained by evaluating expression \lstinline|e|; 
\item \emph{the distance field of the neighbour candidacies is updated},
        (\Cref{fig:algorithm}, Line~\ref{code:protelis:distance}),
        by adding to each candidacy provided by a neighbour the local distance w.r.t. that neighbour\footnote{Note that this operation, together with the \lstinline|share| application, essentially provides the same structure as the basic gradient algorithm discussed in \autoref{example:gradient}.} (as provided by \lstinline|metric|);
\item \emph{neighbours' candidacies that are too far or support the current device are de-prioritised}
        (\Cref{fig:algorithm}, Lines~\ref{code:protelis:validate:call}),
                by function \texttt{\fmpName()}; 
                and, finally
\item \emph{selecting the winner over the processed candidacies through minimisation},
        by \lstinline|min| and \lstinline|foldMin|
        (\Cref{fig:algorithm}, Line~\ref{code:protelis:min}),
        which minimise over the filtered candidacy triplets,
        with default candidacy as the one with the lowest priority (provided by \lstinline|discard()|), and against the \lstinline|local| candidacy.
\end{itemize}
}
\revA{
The above described algorithm is an effective sampling operator (\autoref{def:optsampling}) as long as (i) \emph{half the path sampling error $\samplingErr$ is used as parameter \lstinline|radius|}, so that function \texttt{\fmpName{}} does not expand regions beyond the given error $\samplingErr$, and (ii) \emph{an additive \lstinline|metric| is used}, so that it is impossible to decrease the error by expanding any given region (at best, it will stay the same).
}

\addrev{
\subsection{Formal Analysis}
\label{ssec:formal-results}

In this section, we prove that the proposed solution is self-stabilising (cf. \autoref{def:selfstab-computation}, \autoref{def:selfstab-operator}) and that it represents an effective sampling operator leading to a bounded-error locally optimal regional partition (cf. \autoref{def:sampler}, \autoref{def:optimal}, \autoref{def:optsampling}).

Since our \sampling{} must be a stabilising computation (see \autoref{def:sampler}), we start by proving that our algorithm is self-stabilising.
We do so by exploiting the framework in \cite{TOMACS2018,DBLP:journals/lmcs/AudritoBDPV20},
 which defines a set of self-stabilising \emph{fragments}
 which can be composed together to yield self-stabilising operators (\autoref{def:selfstab-operator}).
\revA{
In particular, in \cite{TOMACS2018} it is proved that any closed expression in the self-stabilising fragment is self-stabilising, by structural induction on the syntax of expressions and programs (\cite{TOMACS2018}, Appendix E, Lemma 2): values and variables are already self-stabilised,
 a function application self-stabilises (by the inductive hypothesis) if its arguments are self-stabilising,
 and similar considerations can be done for the other program fragments.
}
}

\begin{figure}[th]
\centering
\newsavebox{\figstepa}
\begin{lrbox}{\figstepa}
\begin{lstlisting}[basicstyle=\ttfamily\scriptsize]
def updateDistance(x, metric) {
  x.set(1, x.get(1) + metric())
  x
}
def fR(x, prev) = x
def fMP(x, localId, radius, metric) =
  %\fmpName{}%(updateDistance(x, metric), localId, radius)
def minHoodLoc(e, loc) = min(loc, foldMin(discard(), e))
def %\algname{}%(mid, radius, symmetryBreaker, metric) {
    let local = Sample(-symmetryBreaker, 0, mid)
    areaCenter(
        share (x <- local) {
            let candidacies = x.set(1, x.get(1) + metric())
            let filtered = %\fmpName{}%(candidacies, mid, radius)
            min(local, foldMin(discard(), filtered))
        }
    )
}
\end{lstlisting}
\end{lrbox}
\newsavebox{\figstepb}
\begin{lrbox}{\figstepb}
\begin{lstlisting}[basicstyle=\ttfamily\scriptsize]
// ... omitted ...
    areaCenter(
        share (x <- local) {
            let filtered = fMP(x, mid, radius, metric)
            min(local, foldMin(discard(), filtered))
        }
    )
// ... omitted ...
\end{lstlisting}
\end{lrbox}
\newsavebox{\figstepc}
\begin{lrbox}{\figstepc}
\begin{lstlisting}[basicstyle=\ttfamily\scriptsize]
// ... omitted ...
    share (x <- local) {
        let filtered = fMP(x, mid, radius, metric)
        min(local, foldMin(discard(), filtered))
    }.get(2)
// ... omitted ...
\end{lstlisting}
\end{lrbox}
\newsavebox{\figstepd}
\begin{lrbox}{\figstepd}
\begin{lstlisting}[basicstyle=\ttfamily\scriptsize]
// ... omitted ...
    share (x <- local) {
        minHoodLoc(fMP(x, mid, radius, metric), local)
    }.get(2)
// ... omitted ...
\end{lstlisting}
\end{lrbox}

\subfloat[First step of the conversion to a minimising-\lstinline|share| form:
    \lstinline|received| has been renamed to \lstinline|x|;
    functions \lstinline|updateDistance|, \lstinline|fR|, \lstinline|fMP|, and \lstinline|minHoodLoc| have been defined.\label{fig:step1}]{\begin{minipage}{\textwidth}
\usebox{\figstepa}\end{minipage}}

\subfloat[The update to the distance field and the call to
\texttt{\fmpName{}} have been replaced by \lstinline|fMP|.\label{fig:step2}]{\begin{minipage}{\textwidth}\usebox{\figstepb}\end{minipage}}

\subfloat[Replace \lstinline|areaCenter| with its definition (selection of the second element in the tuple).\label{fig:step3}]{\begin{minipage}{\textwidth}\usebox{\figstepc}\end{minipage}}

\subfloat[Replace the \lstinline|share| body with a call to \lstinline|minHoodLoc|.\label{fig:step4}]{\begin{minipage}{\textwidth}\usebox{\figstepd}\end{minipage}}

\caption{Syntactic steps for passing from \Cref{fig:algorithm} to \Cref{code:protelis:forproof}.}
\end{figure}

\begin{figure}
\begin{lstlisting}
def updateDistance(x, metric) {
  x.set(1, x.get(1) + metric())
  x
}
def fR(x, prev) = x // raising function
def fMP(x, localId, radius, metric) = // monotonic progr.
  %\fmpName{}%(updateDistance(x, metric), localId, radius)
def minHoodLoc(e, loc) = // minimum of loc and e's values
  min(loc, foldMin(discard(), e))
def %\algname{}%(mid, radius, symmetryBreaker, metric) {
  let local = Sample(-symmetryBreaker, 0, mid)
  share (x <- local) {
    fR(minHoodLoc(fMP(x, mid, radius, metric), local), x)
  }.get(2)
}
\end{lstlisting}
  \caption{
    Protelis code from \Cref{fig:algorithm} rewritten to conform to the \emph{minimising share} self-stabilising pattern as per the Proof of \autoref{proof:self-stab}.
  }
  \label{code:protelis:forproof}
\end{figure}

\begin{thm}[\algname{} is self-stabilising\label{proof:self-stab}]
\end{thm}
\begin{proof}
\newcommand{\fR}{\ensuremath{\mathit{fR}}}
\newcommand{\share}{\ensuremath{\mathit{share}}}
\newcommand{\fMP}{\ensuremath{\mathit{fMP}}}
\newcommand{\minHoodLoc}{\ensuremath{\mathit{minHoodLoc}}}
\addrev{
In \cite{DBLP:journals/lmcs/AudritoBDPV20},
 it is proved that an expression of the following form (called a \emph{minimising $\share$} pattern) is self-stabilising:
 }
\begin{lstlisting}
share(x <- e) { fR(minHoodLoc(fMP(x), e), x) }
\end{lstlisting}
\addrev{
where (see Section~5.2 in \cite{TOMACS2018} and, especially, Section~4.7 and Figure~7 in \cite{DBLP:journals/lmcs/AudritoBDPV20}):
\begin{itemize}
  \item \lstinline|fR(x, prev)| is a ``raising function'',
  with respect to partial orders,
  of  \lstinline|x| and \lstinline|prev|
  (the value of \lstinline|x| at the previous round);
  \item \lstinline|fMP| is a monotonic progressive function of \lstinline|x|,
  which can take further arguments as far as they are self-stabilising expressions
  that do not contain the \lstinline|share|-bounded variable \lstinline|x|; and 
  \item \lstinline|minHoodLoc(e, loc)| selects the minimum among the neighbours' values of expression \lstinline|e| and the current device's local value \lstinline|loc|.
\end{itemize}
\revA{Now, the block of Protelis code (\Cref{fig:algorithm})
in Lines~\ref{code:protelis:fun:init}--\ref{code:protelis:fun:end}
can be rewritten as shown in \Cref{code:protelis:forproof}
that
conforms to the minimising \texttt{share} pattern, where:
\begin{itemize}
  \item the raising function \lstinline|fR| is an identity on the first parameter,
  which is a trivially valid raising function
  (see Example 5.5 in \cite{TOMACS2018});
  \item function \texttt{\fmpName{}} is a valid monotonic progressive function \lstinline|fMP| of \lstinline|x|,
  since it transforms neighbours' candidacies supporting the current device
  and those at a distance farther than \lstinline|radius|
  to the highest value for the data type (cf. \lstinline|discard()|), leaving the others unaltered;
  none of the provided additional arguments
  (\lstinline|id|, \lstinline|radius|, and \lstinline|metric|)
  contains the \lstinline|share|-bounded variable \lstinline|x|.
\end{itemize}
More gradually, the transformation can be obtained by:
\begin{enumerate}
    \item renaming \lstinline|received| to \lstinline|x|, and defining functions \lstinline|fR|, \lstinline|fMP|, and \lstinline|minHoodLoc| (\Cref{fig:step1});
    \item realising that \lstinline|fMP| is a valid replacement for the combination of distance field update and call to \texttt{\fmpName{}} and replacing accordingly (\Cref{fig:step2});
    \item replacing \lstinline|areaCenter| with its definition (\Cref{fig:step3});
    \item replacing the \lstinline|share| body with \lstinline|minHoodLoc|, as they perform the same operation (\Cref{fig:step4});
    \item adding a call to \lstinline|fR|, which is an identity function, does not alter the behaviour of the code and leads directly to the code in \Cref{code:protelis:forproof}.
\end{enumerate}
}

The other elements in the program
  are only operations on local data 
  which 
  are also self-stabilising expressions.
Since the \algname{} function consists exclusively of self-stabilising expressions,
 it is in turn self-stabilising~\cite{TOMACS2018,DBLP:journals/lmcs/AudritoBDPV20}.
}
\end{proof}

\addrev{
\begin{thm}[\algname{} is an effective sampling operator\label{thm:effective-sampling}]
\end{thm}
\begin{proof}
To prove that our algorithm 
 represents an effective sampling operator $\Program_\samplingErr = \text{\algname{}}$,
 we have to prove that it yields, on any stabilising input field $f_i$, 
 an output stabilising field $f_o$ of locally optimal regional partitions.
As per \autoref{proof:self-stab}, $\Program_\samplingErr$ is self-stabilising, so on a stable input it will yield a stable output:  let $f_o^s$ be its snapshot, and $f_i^s$ the corresponding input.

\revT{On the one hand, accuracy is guaranteed since $\algname$ ensures that no device has \pse{} greater than $\samplingErr/2$ from the leader: for the triangular inequality property of metric spaces ($m(a,b)\leq m(a,c) + m(c,b)$) this ensures that stable snapshot distance $\fmetric$ does not overcome $\samplingErr$.

On the other hand, for local optimality under error $\samplingErr$ and distance $\fmetric$,
 there must not exist two contiguous regions 
 $\eventS',\eventS''\in\Regions(f_o^s)$,
 with samplers $\deviceId'$ and $\deviceId''$, where 
$\fmetric(\trestrict{f_i^s}{\eventS'\cup\eventS''},\trestrict{f_o^s}{\eventS'\cup\eventS''}) \leq \samplingErr/2$.
Suppose two such regions exists, and let $\deviceId'$ be stronger than $\deviceId''$ (i.e., it has higher symmetry breaker). 
Then, the \pse{} between $\deviceId'$ and $\deviceId''$ is necessarily higher than $\samplingErr/2$, because of the steps \ref{algostep-discard-candidacy} and \ref{algostep-select-candidacy}  of the algorithm: in fact, if it were smaller than $\samplingErr/2$ then $\deviceId''$ would have followed $\deviceId'$, and would not have been a leader (cf. \Cref{fig:algoexample-regions-forproof}).}
\end{proof}
}

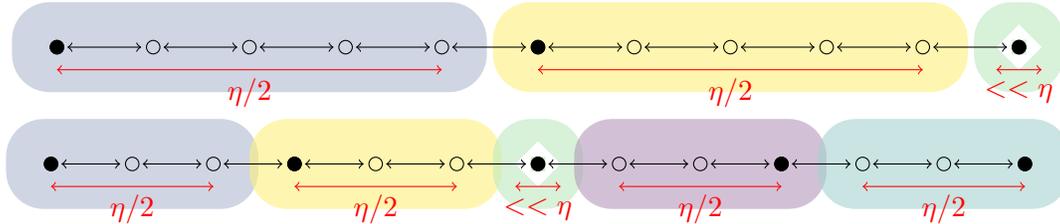
\begin{figure}[h]
\begin{tikzpicture}

\newcommand{\cRegion}[3]{
\begin{scope}[on background layer]
    \filldraw[draw=#1,fill=#1, line join=round, line width=1cm,fill opacity=\colorRegionOpacity,draw opacity=\colorRegionOpacity] plot coordinates{
	(n#2.north west) (n#3.north east)
	(n#3.south east) (n#2.south west)
}--cycle;
\end{scope}
}

\node[dev] (n0) [] {};
\foreach \dn [evaluate=\dn as \prev using \dn-1] in {1,...,10}
  \node[dev] (n\dn) [right=of n\prev] {};

\foreach \dn 
  in {0,...,9}
    \pgfmathsetmacro\succ{\dn + 1}
	\draw[<->] (n\dn) -- ($(n\succ)+(-0.3,0)$);

\cRegion{\colorRegionA}{0}{4}
\cRegion{\colorRegionB}{5}{9}
\cRegion{\colorRegionC}{10}{10}

\foreach \dn in {0,5,10}
  \node[dev,fill=black] (over\dn) at (n\dn) {};

\draw[<->,red] ($(n0)+(0,-0.3)$) -- ($(n4)+(0,-0.3)$) node [midway,yshift=-0.3cm] {$\eta/2$};
\draw[<->,red] ($(n5)+(0,-0.3)$) -- ($(n9)+(0,-0.3)$) node [midway,yshift=-0.3cm] {$\eta/2$};

\draw[<->,red] ($(n10)+(-0.3,-0.3)$) -- ($(n10)+(+0.3,-0.3)$) node [midway,yshift=-0.3cm] {$<< \eta$};

\end{tikzpicture}

\begin{tikzpicture}[node distance=0.8cm]

\newcommand{\cRegion}[3]{
\begin{scope}[on background layer]
    \filldraw[draw=#1,fill=#1, line join=round, line width=1cm,fill opacity=\colorRegionOpacity,draw opacity=\colorRegionOpacity] plot coordinates{
	(n#2.north west) (n#3.north east)
	(n#3.south east) (n#2.south west)
}--cycle;
\end{scope}
}

\node[dev] (n0) [] {};
\foreach \dn [evaluate=\dn as \prev using \dn-1] in {1,...,12}
  \node[dev] (n\dn) [right=of n\prev] {};

\foreach \dn 
  in {0,...,11}
    \pgfmathsetmacro\succ{\dn + 1}
	\draw[<->] (n\dn) -- ($(n\succ)+(-0.3,0)$);

\cRegion{\colorRegionA}{0}{2}
\cRegion{\colorRegionB}{3}{5}
\cRegion{\colorRegionC}{6}{6}
\cRegion{\colorRegionD}{7}{9}
\cRegion{\colorRegionE}{10}{12}

\foreach \dn in {0,3,6,9,12}
  \node[dev,fill=black] (over\dn) at (n\dn) {};

\draw[<->,red] ($(n0)+(0,-0.3)$) -- ($(n2)+(0,-0.3)$) node [midway,yshift=-0.3cm] {$\eta/2$};
\draw[<->,red] ($(n3)+(0,-0.3)$) -- ($(n5)+(0,-0.3)$) node [midway,yshift=-0.3cm] {$\eta/2$};
\draw[<->,red] ($(n12)+(0,-0.3)$) -- ($(n10)+(0,-0.3)$) node [midway,yshift=-0.3cm] {$\eta/2$};
\draw[<->,red] ($(n9)+(0,-0.3)$) -- ($(n7)+(0,-0.3)$) node [midway,yshift=-0.3cm] {$\eta/2$};
\draw[<->,red] ($(n6)+(-0.3,-0.3)$) -- ($(n6)+(+0.3,-0.3)$) node [midway,yshift=-0.3cm] {$<< \eta$};

\end{tikzpicture}

\caption{\revA{
Examples of regional partitionings with efficiency $0.5$. Notation: black dots denote the leaders; the coloured areas denote regions; and the red extension lines are used to denote the error-distances. Note how the union of the green singleton region (associated to the weakest leader) with its neighbouring regions would make the \ssed{} of the latter exceed $\eta/2$---if that would not be the case, then the former region would not have existed in the first place (cf. \autoref{thm:effective-sampling}).}
}
\label{fig:algoexample-regions-forproof}
\end{figure}

\revA{
\subsection{Cost analysis}
\label{ssec:cost-analysis}

Executing one cycle of the proposed algorithm requires operating over the information received from all the neighbours,
which is, of course, proportional to their number.
Operations within the \lstinline|share| block
(Lines \ref{code:protelis:distance}--\ref{code:protelis:min} in \Cref{fig:algorithm})
are, indeed, operations on fields:
by the aggregate computing semantics, they are evaluated for each neighbour.
Thus, computationally, the cost of the algorithm is proportional to neighbourhood size: 
larger neighbourhoods require more effort.

From the point of view of message size,
the payload of the algorithm has two components:
the data type used to represent the sample, and, potentially,
additional data that needs to be shared to compute the result of the \lstinline|metric| function.
The former depends on the actual types used for \lstinline|symmetryBreaker| and \lstinline|leaderId|,
the latter on the type of distance returned by \lstinline|metric|.

For example, assuming a classic TCP/IPv6 network and devices with a single network interface,
we could use the local MAC address (6 bytes) as \lstinline|symmetryBreaker|,
the IP address (16 bytes) as \lstinline|leaderId|,
and a 4-byte floating-point number as return type of \lstinline|metric|,
resulting in a payload of 26 bytes per device per round.
Additionally, however,
further data might have to be shared to compute the \lstinline|metric|;
for instance, if devices are equipped with a GPS, they may compute distances by sharing their coordinates
and using the Haversine algorithm.
Assuming a local sensor named \lstinline|gps|, the Protelis code for such implementation of \lstinline|metric| could be:
}
\begin{lstlisting}[language=Protelis]
def distanceWithGps() {
  let latLong = env.get("gps")
  haversine(latLong, nbr(latLong))
}
\end{lstlisting}
\revA{
The \lstinline|nbr| call would incur into an additional network cost, as the local position would be shared with all neighbours.
Assuming a couple of 4-bytes floating-point numbers for latitude and longitude,
that would result into an additional 8 bytes per device per round,
bringing the total up to 34 bytes from the initial 26 bytes.
Notice, however, that this additional cost could get nullified if the \lstinline|metric| function is implemented
in a way that does not require additional data to be shared
(for instance, by using the wireless signal strength as a proxy for the distance, or by using the hop distance).

If the algorithm is implemented in an aggregate computing language,
and no application-specific optimisation is devised,
an identifier for each interaction (\lstinline|share| or \lstinline|nbr| call) is attached to the message,
so the payload would also include the size of one or two identifiers
(again, depending on whether the \lstinline|metric| requires data to be shared).
}

\section{Evaluation}
\label{s:eval}

\addrev{
This section discusses the evaluation of the algorithm proposed in \Cref{s:solution}
against the properties defined in \Cref{s:problem}
by means of simulation.
We first present 
 the evaluation goals (\Cref{ssec:eval-goals}),
 scenarios (\Cref{scenarios}),
 parameters (\Cref{ssec:eval-params}),
 evaluation metrics (\Cref{ssec:eval-metrics}),
 and main implementation details (\Cref{ssec:eval-impl}),
 and finally provide a discussion of the results (\Cref{ssec:eval-results}).
The whole experimental framework has been published as
permanently available artefacts~\cite{zenodo-space-fluid-experiments,zenodo-space-fluid-experiments2} in Zenodo,
with instructions for replicating the results.
}

\subsection{Evaluation Goals}
\label{ssec:eval-goals}

In this section,
we validate the behaviour of the proposed effective aggregate sampling algorithm.
The goals of the evaluation are the following:
\begin{itemize}
	\item \emph{stabilisation}:
    we expect the algorithm to be \emph{self-stabilising}
    (as per \autoref{def:selfstab-operator} and \autoref{proof:self-stab}),
    and thus to behave in a self-stabilising way under different conditions;
	\item \emph{high information (entropy)}:
    we expect the algorithm to split areas with different measurements,
    namely, to dynamically increase the number of regions on a per-need basis
    to minimise the \emph{aggregate sampling error} (as per \autoref{def:error});
	\item \emph{error-controlled upscaling}:
    we expect the algorithm to not abuse of region creation,
    but to keep the minimum number of regions (hence of the largest possible size\revA{---efficiency}) required to maintain accuracy
    (as per \autoref{def:optimal} and \autoref{thm:effective-sampling}),
    intuitively, grouping together devices with similar measurements.
\end{itemize}
Clearly, upscaling and high information density are at odds:
maximum information is achieved by maximising the number of regions,
and thus assigning each device a unique region;
however, doing so would prevent any upscaling.
On the other hand, the maximum possible upscaling 
would be achieved when all devices belong to the same region,
thus minimising information.
We want our regions to change in space ``fluidly'' and opportunistically tracking the situation at hand,
achieving a trade-off between upscaling and amount of information (as per \autoref{def:optsampling}).

\subsection{Scenarios}
\label{scenarios}

We challenge the proposed approach by letting the algorithm operate on synthetic and realistic scenarios.

In the synthetic scenarios, we use different deployments of one thousand devices and different data sources.
We deploy devices into a square arena with different topologies:
\begin{enumerate}[label=\emph{\roman*)}]
	\item \emph{grid (regular grid)}: devices are regularly located in a grid;
	\item \emph{pgrid (perturbed/irregular grid)}: starting from a grid,
	devices' positions are perturbed randomly on both axes;
	\item \emph{uniform}: positions are generated with a uniform random distribution;
	\item \emph{exp (exponential random)}: positions are generated with a uniform random distribution on one axis
	and with an exponential distribution on the other,
	thus challenging device-distribution sensitivity.
\end{enumerate}
In all cases, we avoid network segmentation by forcing each device to communicate \emph{at least} with the eight closest devices.
We simulate the system when sampling the following phenomena:
\begin{enumerate}[label=\emph{\roman*)}]
	\item \emph{Constant}:
	the signal is the same across the space,
	we expect the system to upscale as much as possible;
	\item \emph{Uniform}:
	the signal has maximum entropy,
	each point in space has a random value,
	we thus expect the system to create many small regions;
	\item \emph{Bivariate Gaussian (gauss)}:
	the signal has higher value at the centre of the network,
	and lower towards the borders,
	producing a Gaussian curve whose expected value is located at the centre of the network,
	we expect regions to be smaller where the data changes more quickly;
	\item \emph{Multiple bivariate Gaussian (multi-gauss)}:
	similar to the previous case,
	but the signal value is built by summing three bivariate Gaussian whose expected value is one third
	of the previous Gaussian,
	and whose expected values are located along the diagonal of the network
	(bottom-left corner, centre, top-right corner);
	\item \emph{Dynamic}: the system cycles across the previous states,
	we use this configuration to investigate whether and how the proposed solution adapts to changes in the structure of the signal.
\end{enumerate}

\revA{In the realistic scenario,
we use air quality data from the European Environment Agency~\cite{pm10},
and, specifically, the PM10 data from February 2020 (included) to May 2020 (excluded).
We position the sensor stations in their correct position as reported by the agency,
and assume logical connectivity with close-by stations.
We force each station to communicate with at least the closest stations,
and we ensure that no network segmentation exists by enforcing full network reachability.
This results in a much sparser network than the synthetic ones,
and whose variance in the number of neighbours is much higher:
some stations located in places far from geographical Europe
(such as Réunion and other French overseas departments) have very few connections (possibly, a single one),
while sensors located in dense urban areas can have dozens.
To emulate energy-constrained devices, such as LoRaWAN motes,
we limit the operating frequency of each device to $\nicefrac{1}{1800}$Hz
(namely, one round every half hour on average).
}

\subsection{Parameters}
\label{ssec:eval-params}

The proposed solution can be tuned by three main parameters:
the leader strength, the error tolerance, and the distance metric.
In the experiments, we fix the error tolerance to a constant value,
while we choose among three different alternatives for the leader strength and the distance metric.

\revA{For the leader strength parameter
we use the local concentration of $PM_{10}$ in the realistic experiment, 
while in the synthetic one we consider:}
\begin{enumerate}[label=\emph{\roman*)}]
	\item \emph{value}: the local value of the tracked signal $s$;
	\item \emph{mean}: the neighbourhood-mean value of the tracked signal $s$,
    assuming $N$ to be the set of neighbours (including the local device),
    and $s_i$ to be the value of the tracked signal at device $i\in N$,
    the value is computed as:
    
    $$M=\frac{\sum\limits_{i \in N}{s_i}}{|N|}$$
    
	\item \emph{variance}: the neighbourhood-variance of the tracked signal $s$,
    assuming $M_i$ to be the neighbourhood-mean computed at device $i\in N$,
    the value is computed as:
    $$\frac{\sum\limits_{i \in N}(M_i-s_i)^2}{|N|}$$
\end{enumerate}

\newcommand{\pmten}{$PM_{10}$}

\begin{figure}
    \begin{center}
    \captionsetup[subfigure]{width=0.3\textwidth}
    \subfloat[\pmten{} concentrations (darker areas have higher \pmten).]{
        \includegraphics[width=0.32\linewidth]{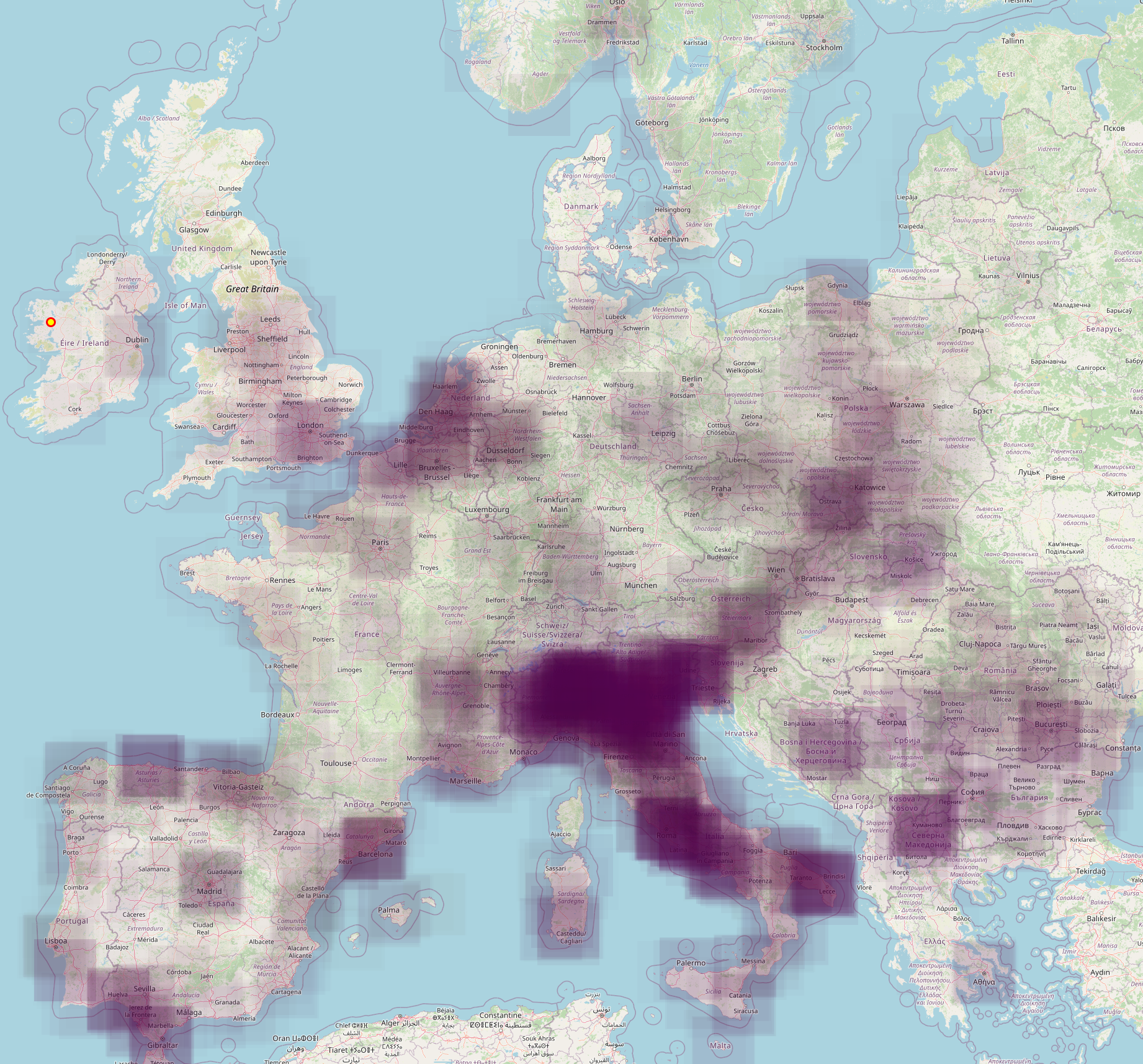}
     }
     \subfloat[Partitions computed using $dist$.]{
        \includegraphics[width=0.32\linewidth]{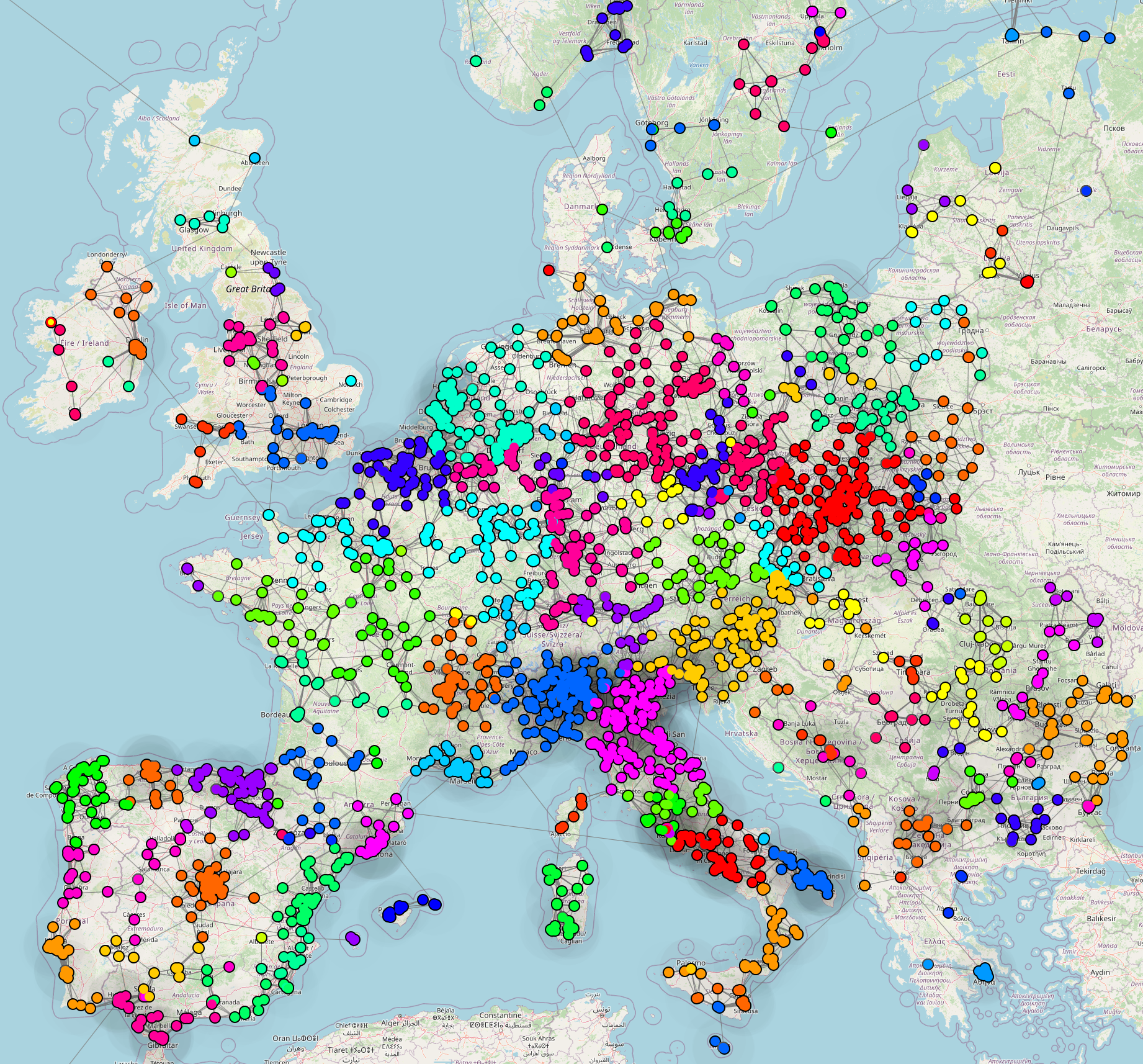}
     }
     \subfloat[Partitions computed using $dist_B$.]{
        \includegraphics[width=0.32\linewidth]{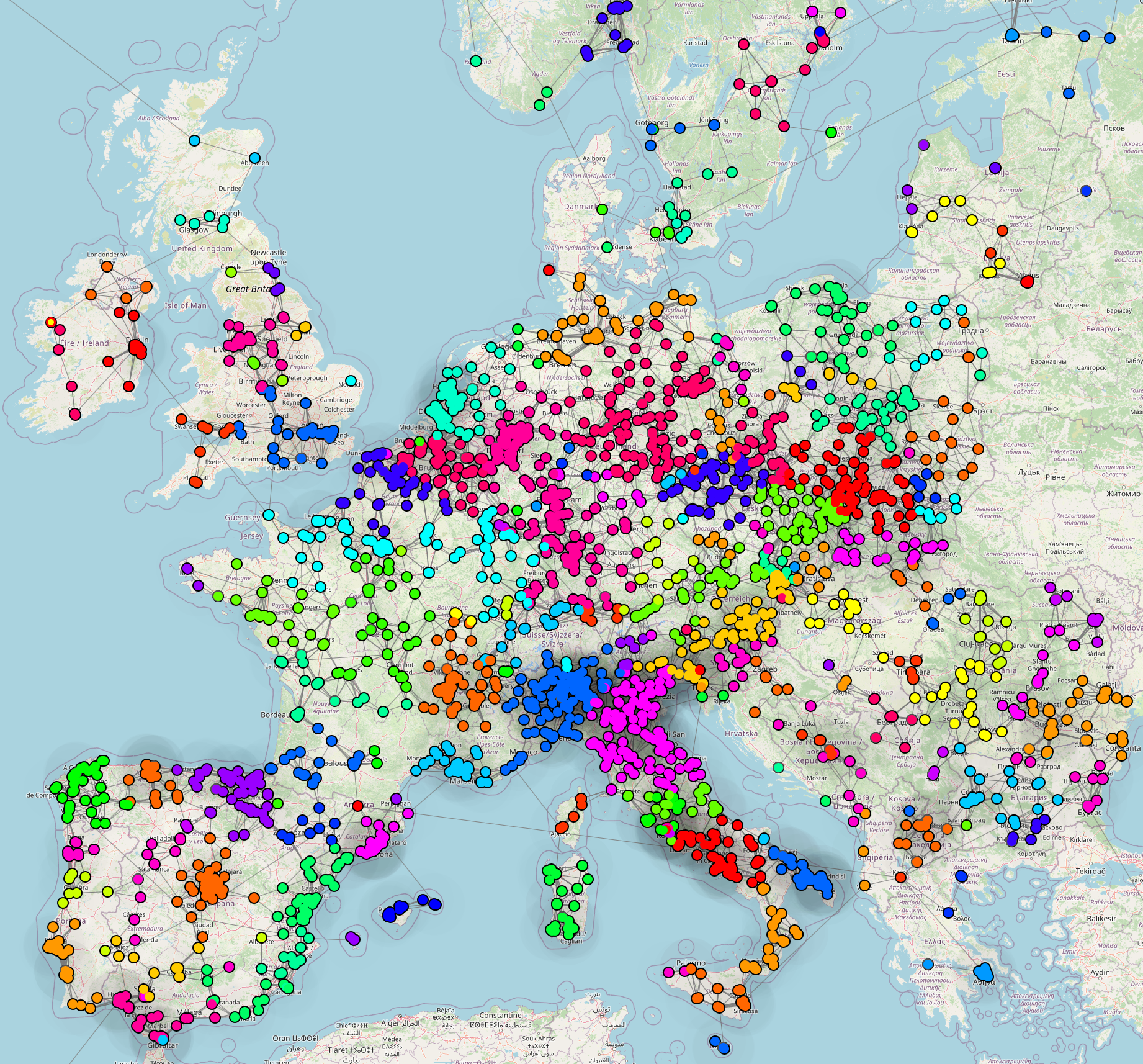}
     }\\
     \subfloat[Partitions computed using $\sigma(PM_{10})$.]{     
        \includegraphics[width=0.32\linewidth]{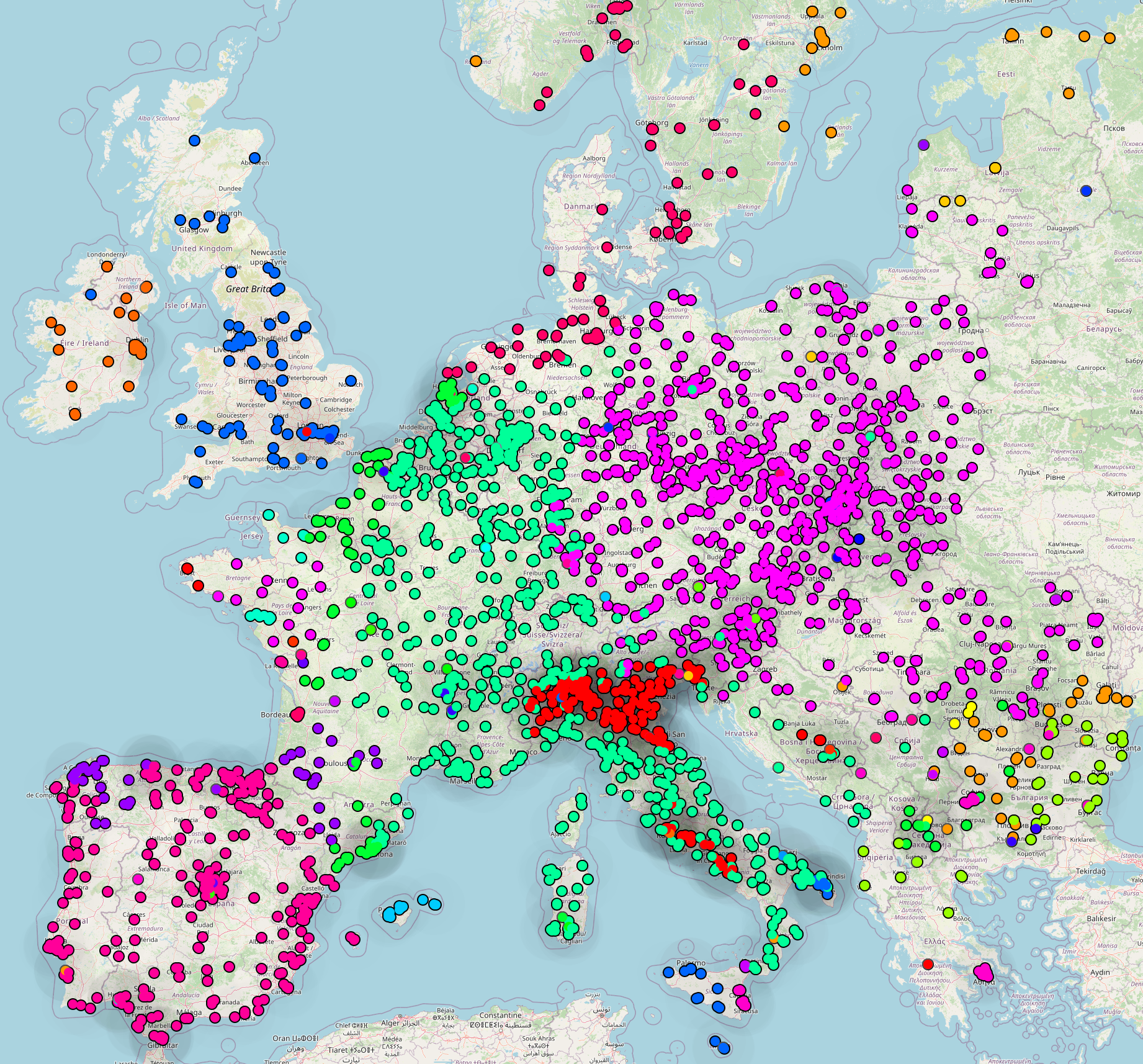}
     }
     \subfloat[Partitions computed using $\sigma(PM_{10})_B$.]{ 
        \includegraphics[width=0.32\linewidth]{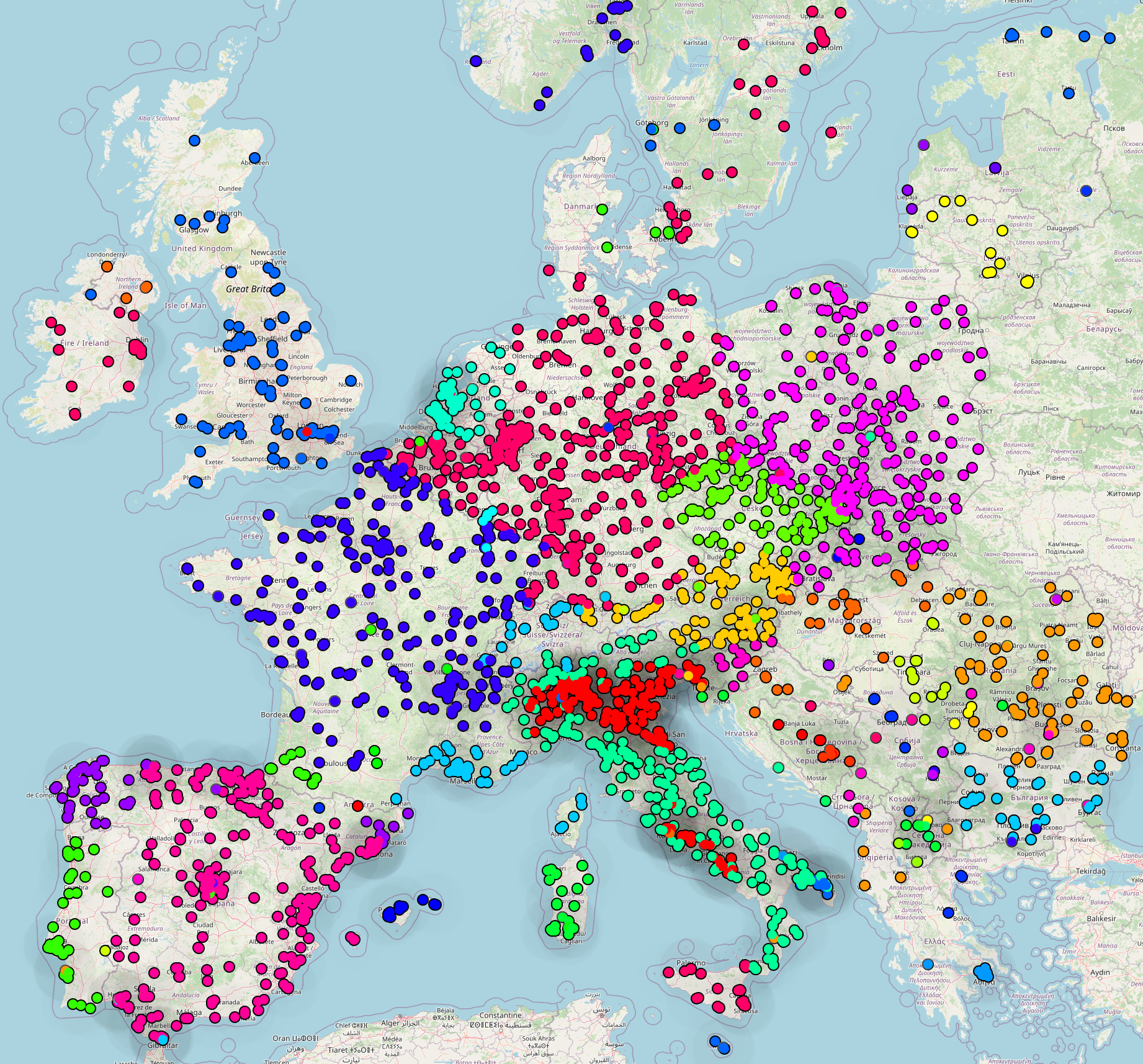}
     }
        \caption{
            \revA{Snapshots produced by the simulator at the same moment.
            Notice that the latter metrics ($\sigma(PM_{10})$ and $\sigma(PM_{10})_B$) are capable of capturing and managing contiguous areas with similar levels of air quality,
            tracking the underlying spatial structure of the signal.}
        }
        \label{fig:snapshots}
    \end{center}
\end{figure}

For the metric, in the synthetic scenario, we consider:
\begin{enumerate}[label=\emph{\roman*)}]
	\item \emph{distance}: the spatial distance is used as distance metric;
	\item \emph{diff}: assuming that $s_i$ is the value of the tracked signal at device $i$,
    the distance between two neighbouring devices $a$ and $b$ is measured as:
    $$e_{ab}=e_{ba}=min(\epsilon, |s_a - s_b|)$$
    where $\epsilon \in \mathbb{R_+}$, $\epsilon = 0$ iff $a=b$, $0 < \epsilon \ll 1$ otherwise,
    we bound the minimum value to preserve the triangle inequality;
	\item \emph{mix}: we mix the two previous metrics
	so that both the error and the physical distance affect in the distance definition;
    i.e., assuming $\overline{ab}$ to be the spatial distance between devices $a$ and $b$,
    we measure the mix metric as:
    $$\overline{ab}\cdot{}e_{ab}$$
\end{enumerate}
\revA{For the realistic scenarios we use instead:
\begin{enumerate}[label=\emph{\roman*)}]
    \item $dist$: the spatial distance is used as distance metric;
    \item $dist_B$: same as $dist$, but country borders are considered as barriers;
    \item $\sigma(PM_{10})$: we weight the distance between neighbouring devices by a factor that depends on
    the Air Quality Index (AQI) value at the device location---devices with more different AQIs are considered more distant;
    \item $\sigma(PM_{10})_B$: same as $\sigma(PM_{10})$, but country borders are considered as barriers.
\end{enumerate}
The idea is to challenge the algorithm by looking at how it behaves when operating on a network with a sparser
and more heterogeneous structure, as well as to investigate the impact of arbitrary limits
and non-linearities (e.g., the country borders) unrelated with the underlying signal included in the expansion metrics.}
\revA{\Cref{fig:snapshots} shows a snapshot in time of the partitioning generated by the simulator in the realistic scenario
for each of the aforementioned metrics.}

\begin{figure}[!b]
	\centering
	\includegraphics[width=\imgfactor\columnwidth]{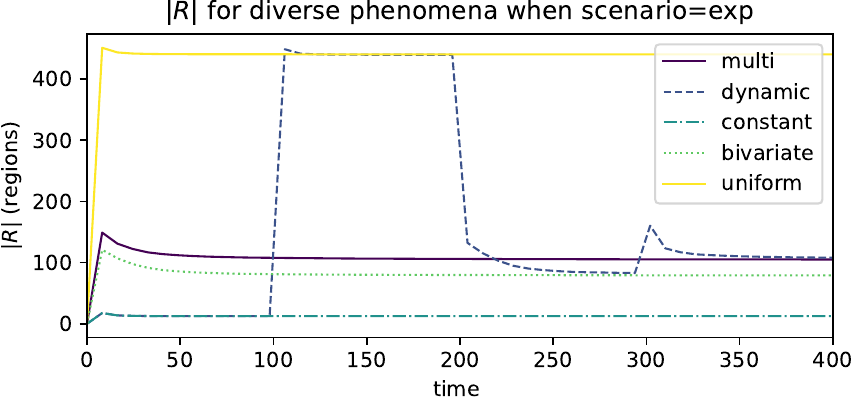}
	\includegraphics[width=\imgfactor\columnwidth]{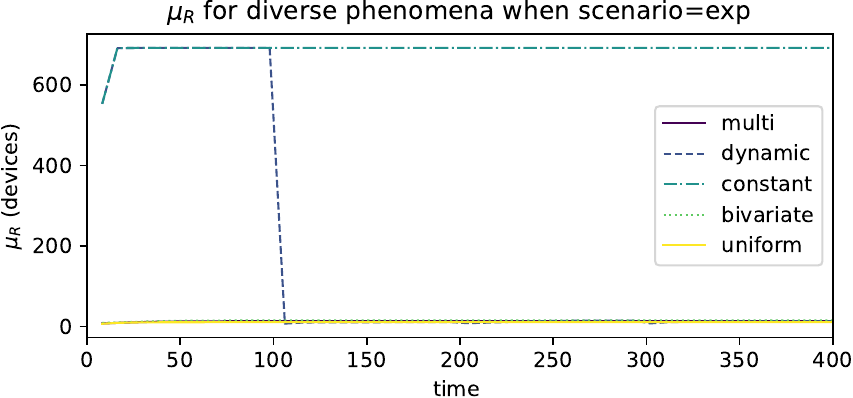}
	\includegraphics[width=\imgfactor\columnwidth]{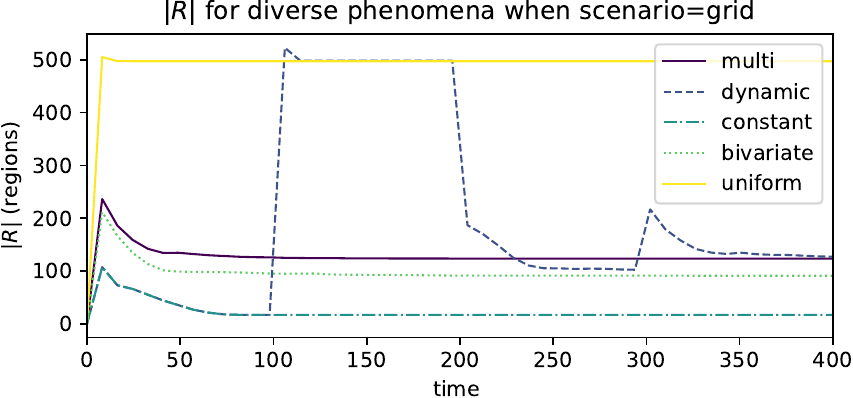}
	\includegraphics[width=\imgfactor\columnwidth]{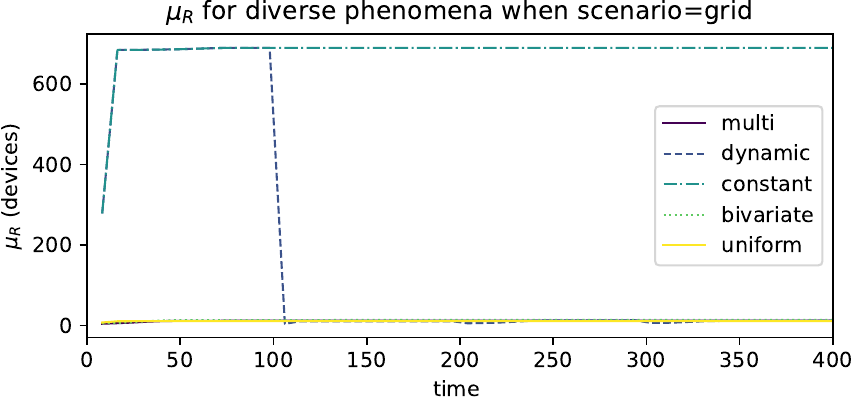}
	\includegraphics[width=\imgfactor\columnwidth]{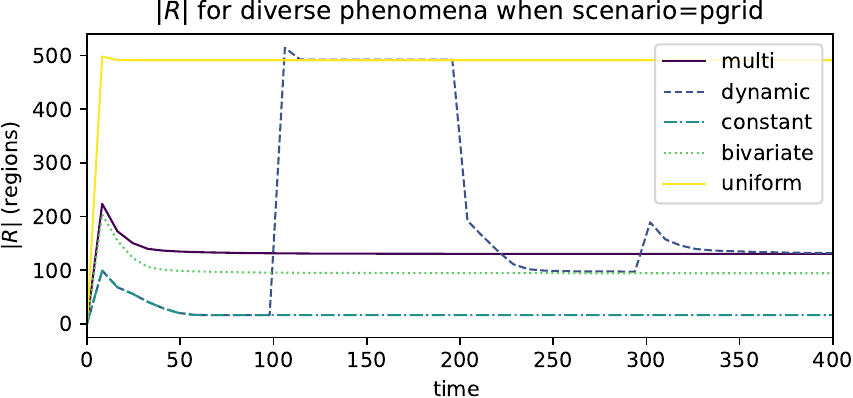}
	\includegraphics[width=\imgfactor\columnwidth]{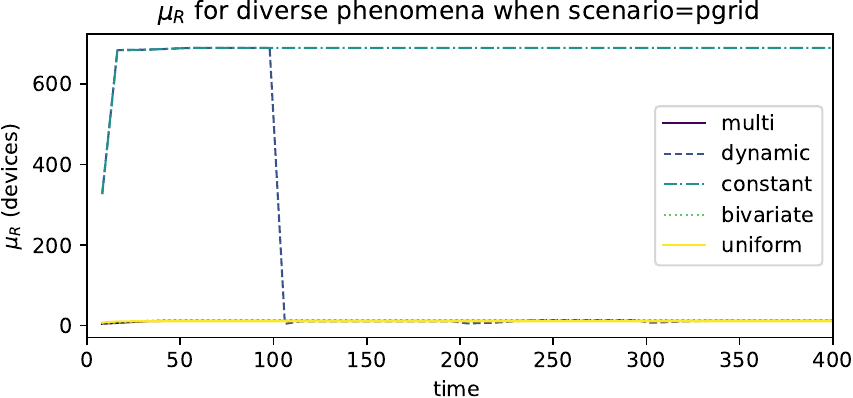}
	\includegraphics[width=\imgfactor\columnwidth]{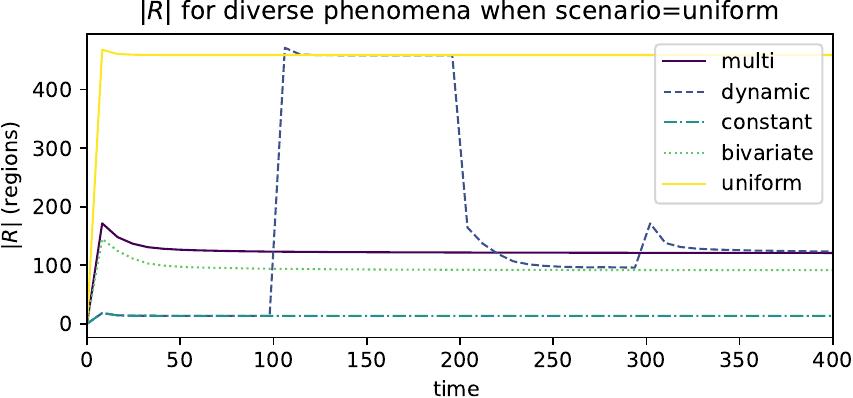}
	\includegraphics[width=\imgfactor\columnwidth]{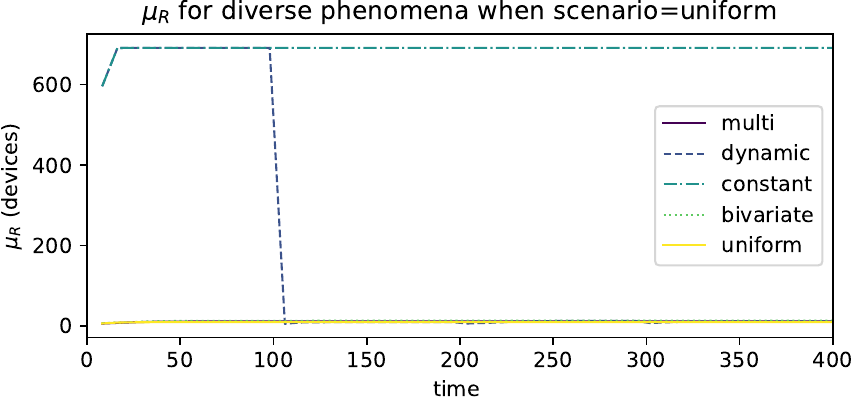}
	\caption{
		Region count (left column) and size (right column) across deployments and scenarios.
		The system behaves very similarly regardless of the device disposition.
		As expected, the higher information density leads to a larger number of smaller regions.
		The dynamic scenario shows that the partitions change in response to changes in the signal.
	}
	\label{fig:region-count-deployment}
\end{figure}

\begin{figure}[!h]
	\centering
	\includegraphics[width=\imgfactor\columnwidth]{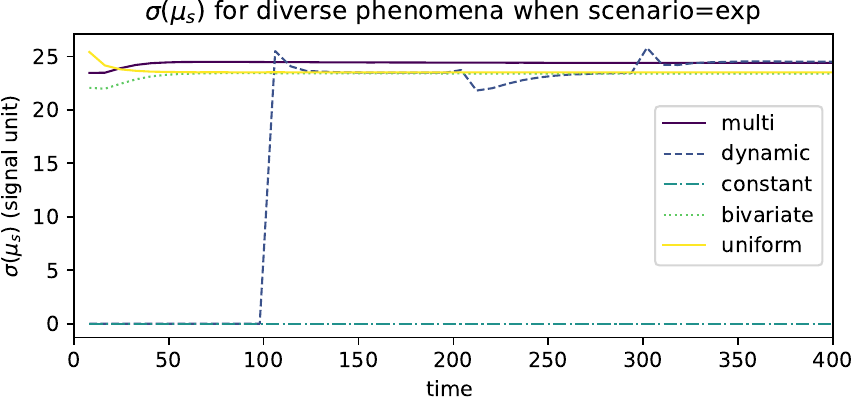}
	\includegraphics[width=\imgfactor\columnwidth]{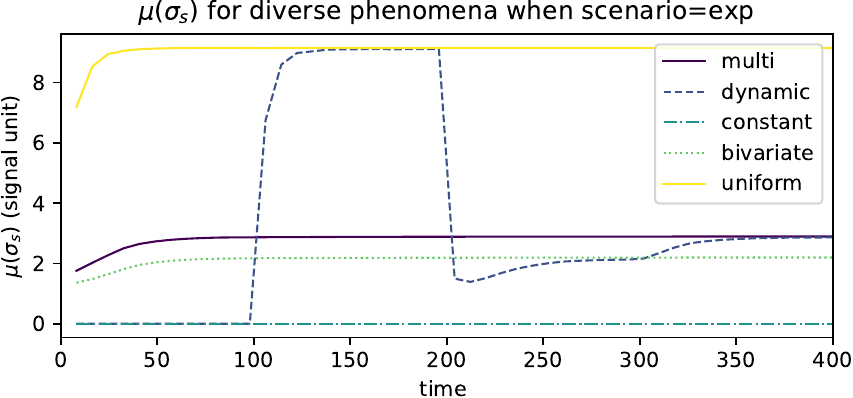}
	\includegraphics[width=\imgfactor\columnwidth]{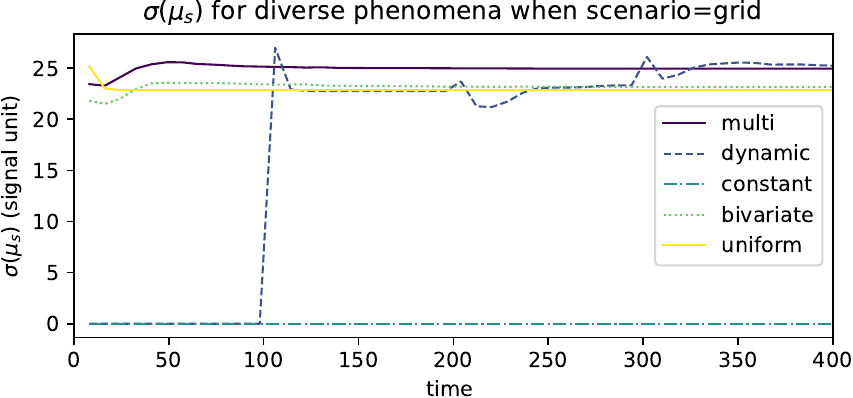}
	\includegraphics[width=\imgfactor\columnwidth]{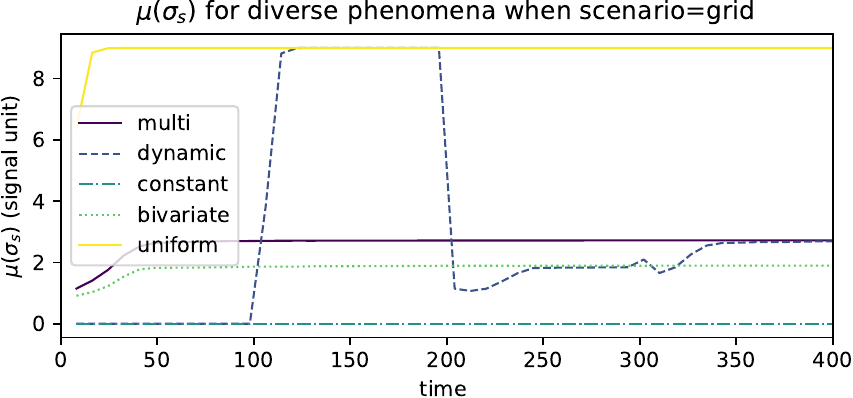}
	\includegraphics[width=\imgfactor\columnwidth]{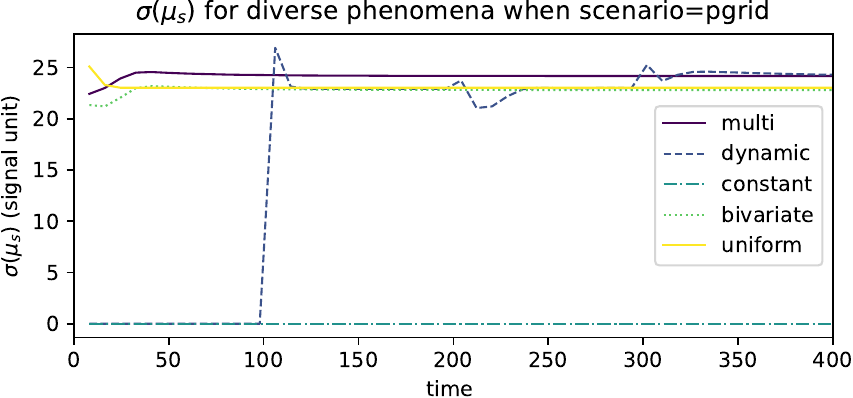}
	\includegraphics[width=\imgfactor\columnwidth]{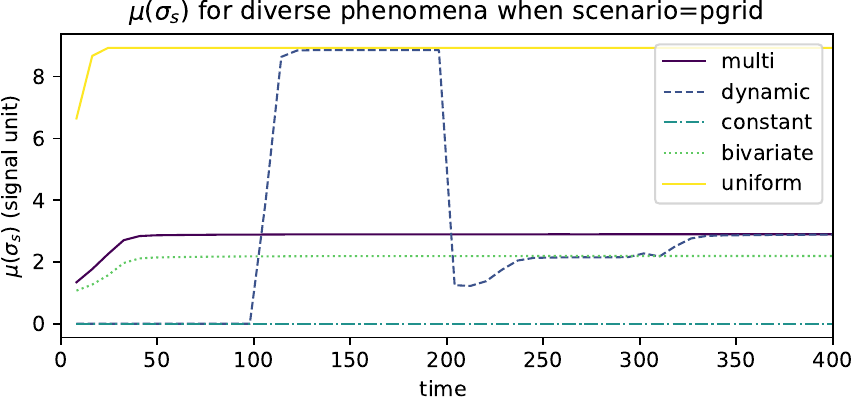}
	\includegraphics[width=\imgfactor\columnwidth]{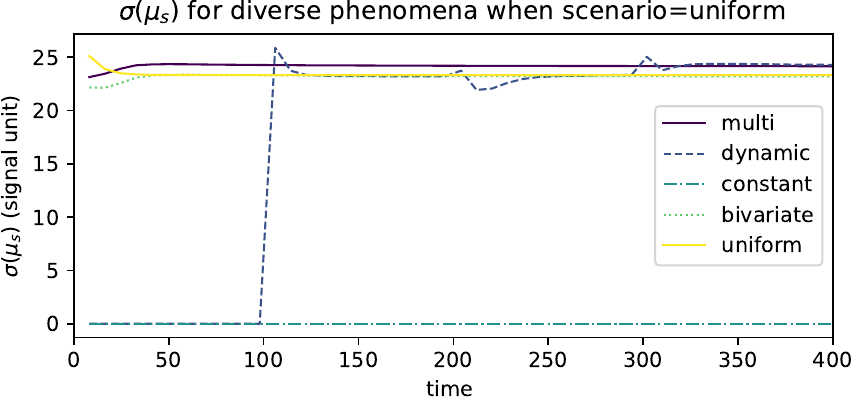}
	\includegraphics[width=\imgfactor\columnwidth]{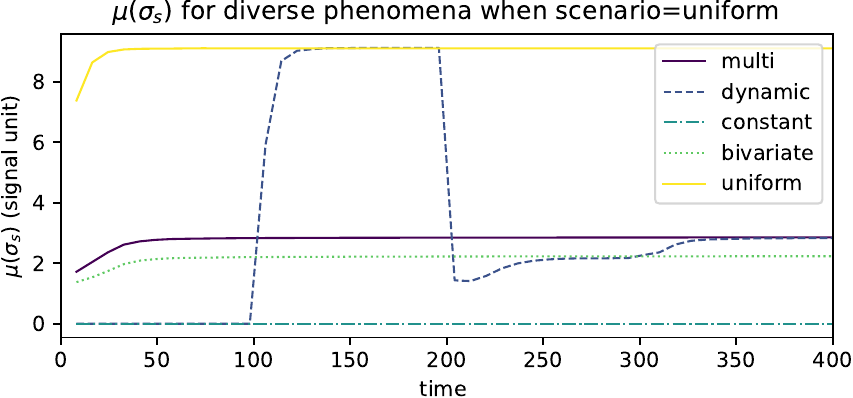}
	\caption{
		Standard deviation of the mean region value (left)
		and mean standard deviation (right)
		across deployments and scenarios,
		indicating respectively how much the regions readings differ from each other
		(the higher the more different)
		and how the regions are internally similar
		(the lower the more homogeneous are regions).
		The constant and uniform random signals work as baselines: in the former case,
		very large areas gets formed, while in the latter most regions count a single device (as expected).
		In the other cases, inter-region differences is maximised
		(they get as high as the most extreme case)
		keeping internal consistency under control.
	}
	\label{fig:intra-inter-deployment}
\end{figure}

\subsection{Evaluation Metrics}
\label{ssec:eval-metrics}

We evaluate the system behaviour by considering the following evaluation metrics.
Assume, at any time instant,
 that a set $D$ of devices is partitioned into a set of regions
$R = R_1 \cup \dots \cup R_{|R|}$,
where each
$R_r$ is a set of devices $\{ D^{r}_1,\ \dots,\ D^{r}_{|R_r|} \}$,
and each device $D^r_d$ senses the local value of the tracked signal $s^r_d$:
\begin{itemize}
    \item \emph{Region count $|R|$ (regions).}
Counting the regions provides an indication about efficiency (\autoref{def:optimal}): more partitions should be expected in environments where the sampled signal has higher entropy.
    \item \emph{Mean region size $\mu_{R} = \nicefrac{\sum_{i=1}^{|R|}{|R_i|}}{|R|}$ (devices).}
Related to the region count, but density-sensitive: when devices are distributed irregularly (as in the \textbf{exp} deployment, see \Cref{scenarios}), we expect this metric to be less predictable.

    \item \emph{Standard deviation of the mean of the signal in regions
$\sigma(\mu_s)$
(same unit of the signal).}
This is a proxy for inter-region difference, with higher values denoting larger differences between different regions.
The mean signal inside region $R_r$ is computed as:
$$\mu_s^{R_r}=\frac{\sum\limits_{i=1}^{|R_r|}{s^r_i}}{|R_r|}$$
while the mean of the means of the signal is

$$\mu_s^{R}=\frac{\sum_{i=1}^{|R|}{\mu_s^{R_i}}}{|R|}$$

thus

$$\sigma(\mu_s) = \sqrt{\frac{1}{|R|}\sum\limits_{i=1}^{|R|}(\mu_s^{R_i} - \mu_s^R)^2}$$

    \item \emph{Mean standard deviation of the signal in regions
$\mu(\sigma_s)$
(same unit of the signal).}
This is a proxy for the intra-region error.
The lower this value, the more similar are the signal readings inside regions,
hence the lower the error induced by the grouping (\autoref{def:error}).
The standard deviation of the tracked signal inside region $R_r$ is computed as:
$$\sigma_s^{R_r} = \sqrt{\frac{1}{|R_r|}\sum\limits_{i=1}^{|R_r|}(s^r_i - \mu_s^{R_r})^2}$$
thus
$$\mu(\sigma_s) = \frac{\sum\limits_{i=1}^{|R|}{\sigma_s^{R_i}}}{|R|}$$

    \item \emph{Standard deviation of the standard deviation of the signal in regions
$\sigma(\sigma_s)$
(same unit of the signal).}
Proxy metric for the consistency of partitioning.
Higher values suggest that partitions have different internal error,
hence behave differently (striving to satisfy \autoref{def:optsampling}).
It is computed as:
$$\sigma(\sigma_s) = \sqrt{\frac{1}{|R|}\sum_{i=1}^{|R|}(\sigma_s^{R_r} - \mu(\sigma_s))^2}$$

\end{itemize}

\subsection{Implementation and Reproducibility}
\label{ssec:eval-impl}

We rely on an implementation coded in the Protelis aggregate programming language~\cite{PianiniSAC2015}.
The simulations are implemented in the Alchemist simulator~\cite{PianiniJOS2013}.
The data analysis leverages Xarray~\cite{xarray} and matplotlib~\cite{matplotlib}.

In the synthetic scenarios,
for each element in the Cartesian product of the
device deployment type,
signal form,
leader strength,
and distance metric,
an experiment was carried out.
Each experiment has been repeated 100 times with different random seeds, resulting in multiple simulation runs per experiment.
Random seeds control both the evolution of the system
(i.e., the order in which devices compute)
and their position on the arena
(except for the regular grid deployment, which is not randomised).

\revA{For the realistic scenario,
since the position of the devices is mandated by the real-world deployment,
we run 10 simulation repetitions for each experiment;
in this case, the random seed controls the evolution of the system
(the order in which the devices compute).}

The presented results are obtained from taking the average of the metrics across all the repetitions;
when a chart does not mention some parameters,
then the results that are presented are also
averaged across all values the parameter may assume for the simulation set.
The experiment has been open sourced,
publicly released\footnote{
    \url{https://github.com/DanySK/Experiment-2022-Coordination-Space-Fluid}
}\footnote{
    \url{https://github.com/DanySK/experiment-2023-lmcs-pm10-pollution-space-sampling}
},
documented,
equipped with a continuous integration system to guarantee replicability,
and published as a permanently available, reusable artefact~\cite{zenodo-space-fluid-experiments}.

\begin{figure}[!h]
	\centering
	\includegraphics[width=\imgfactor\columnwidth]{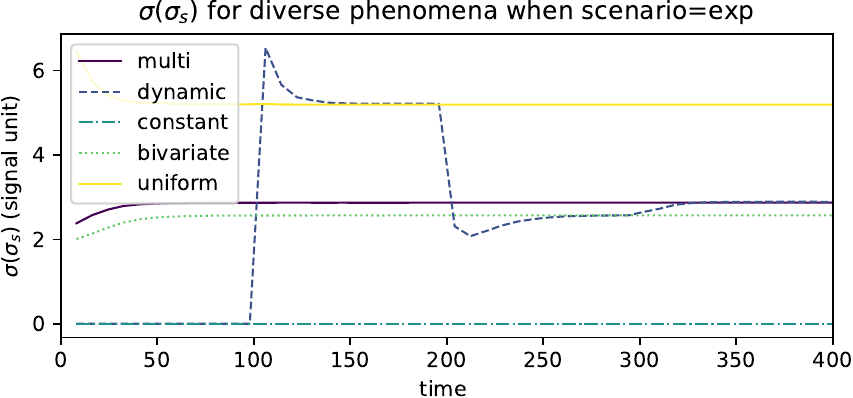}
	\includegraphics[width=\imgfactor\columnwidth]{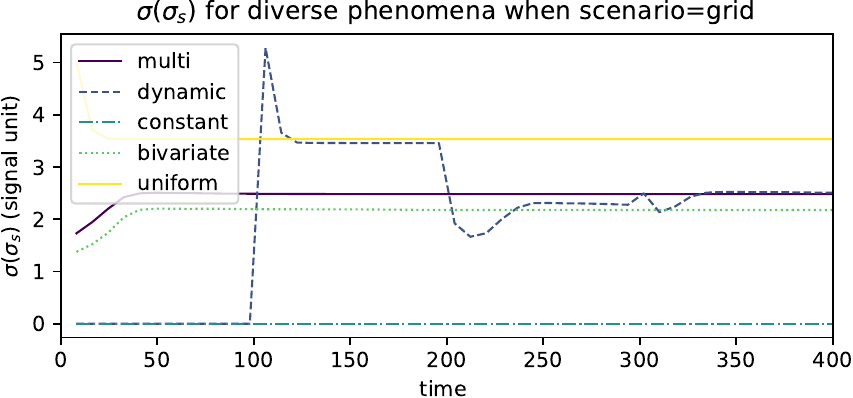}
	\includegraphics[width=\imgfactor\columnwidth]{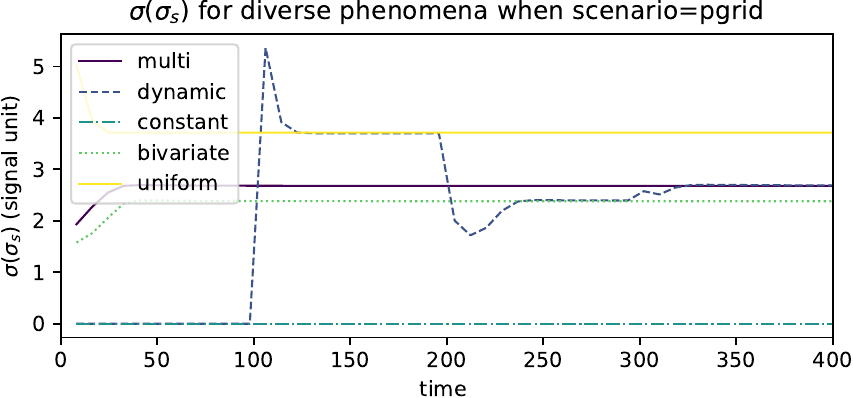}
	\includegraphics[width=\imgfactor\columnwidth]{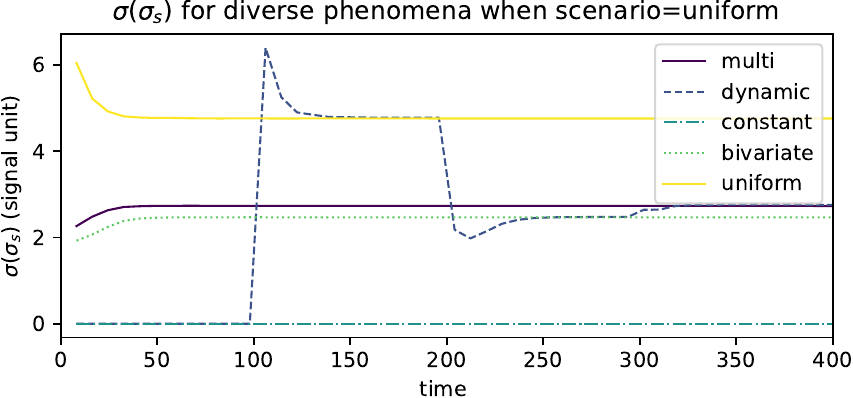}
	\caption{
		Intra-region partitioning homogeneity,
		measured as the standard deviation across regions of the standard deviation of the signal inside regions.
		Higher values denote that different regions are more heterogeneous,
		i.e., that some have larger errors than others.
	}
	\label{fig:intra-deployment}
\end{figure}

\begin{figure}[!h]
	\centering
	\includegraphics[width=\imgfactor\columnwidth]{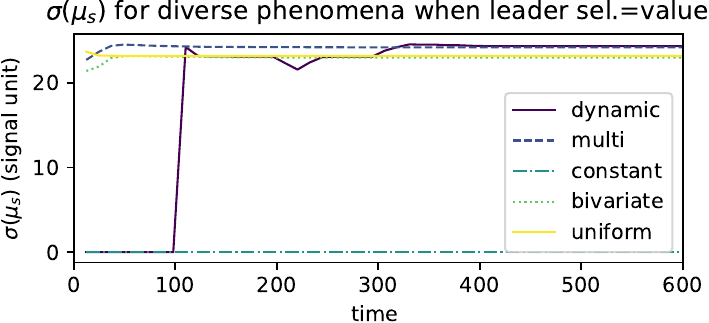}
	\includegraphics[width=\imgfactor\columnwidth]{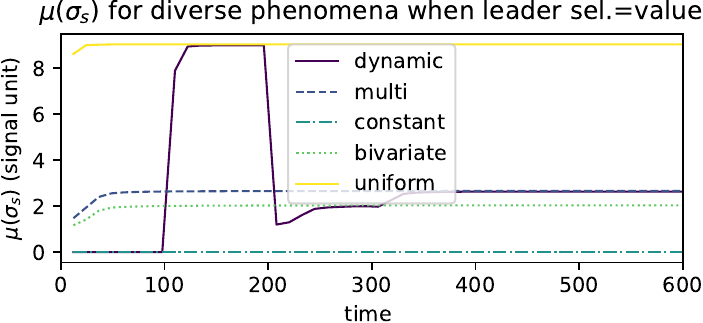}
	\includegraphics[width=\imgfactor\columnwidth]{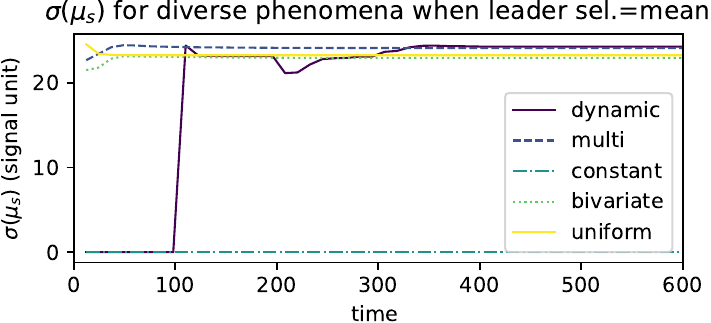}
	\includegraphics[width=\imgfactor\columnwidth]{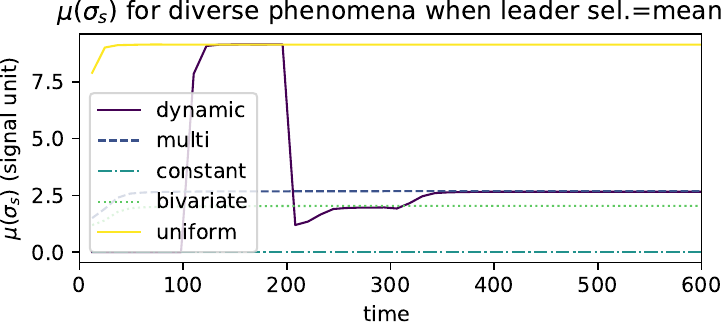}
	\includegraphics[width=\imgfactor\columnwidth]{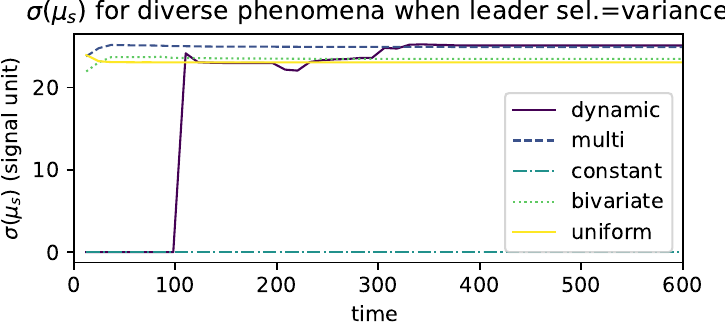}
	\includegraphics[width=\imgfactor\columnwidth]{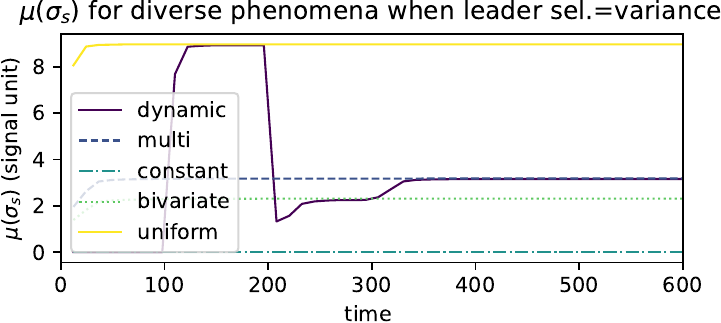}
	\caption{
		Effect of different policies for leader selection.
		The behaviour of the system is similar regardless of the way the region leader is selected.
	}
	\label{fig:leader-selection}
\end{figure}

\begin{figure}[!h]
	\centering
	\includegraphics[width=\imgfactor\columnwidth]{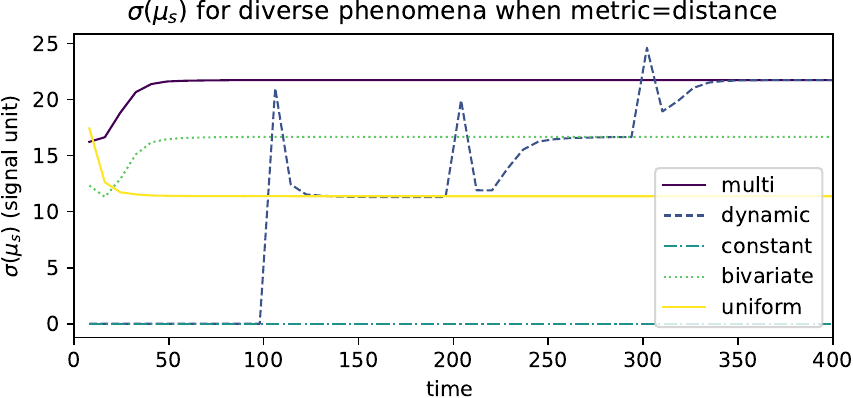}
	\includegraphics[width=\imgfactor\columnwidth]{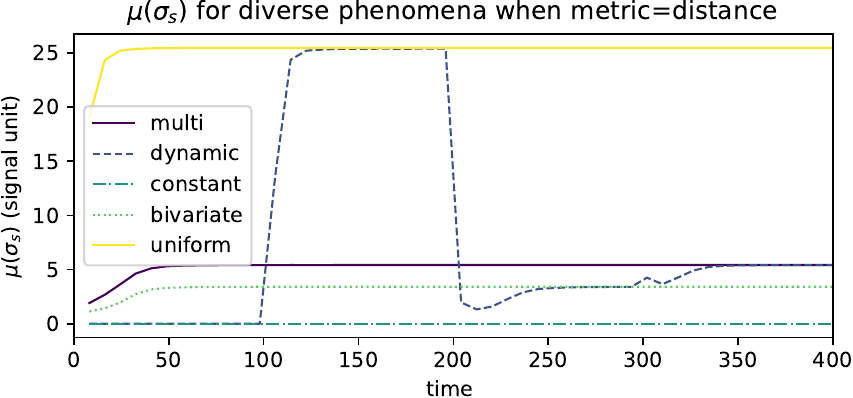}
	\includegraphics[width=\imgfactor\columnwidth]{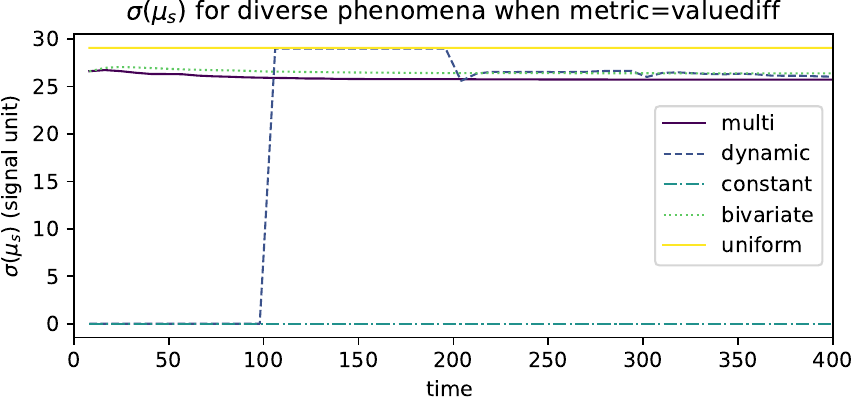}
	\includegraphics[width=\imgfactor\columnwidth]{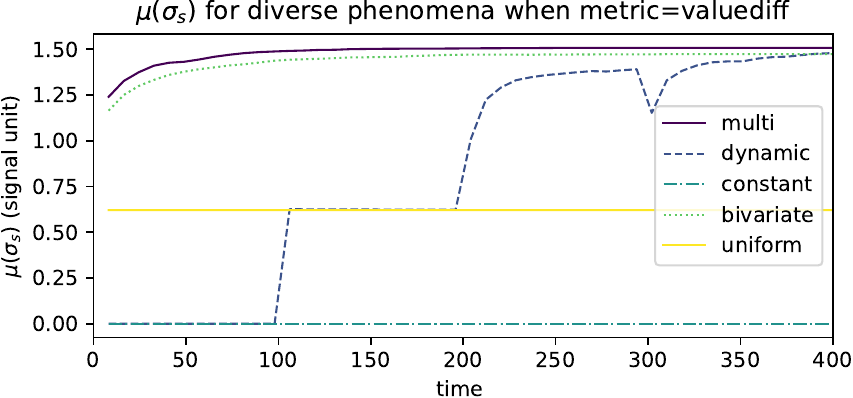}
	\includegraphics[width=\imgfactor\columnwidth]{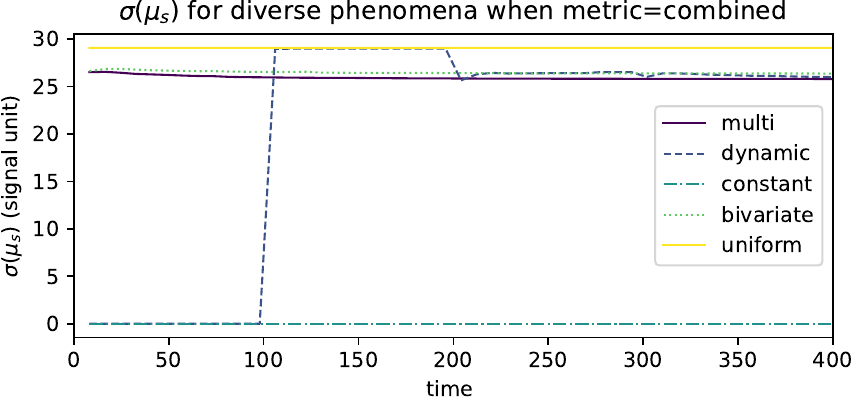}
	\includegraphics[width=\imgfactor\columnwidth]{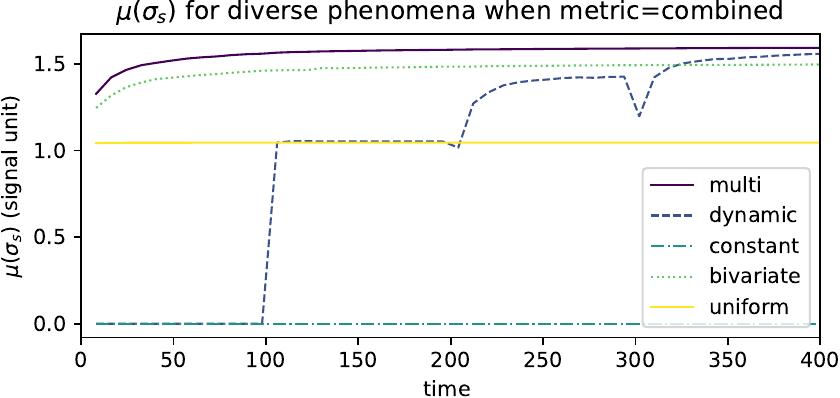}
	\caption{
		Effect of different error measurement metrics.
		The system is very sensible to the metric used to accumulate error,
		which directly impacts the way distance is perceived,
		thus determining the maximum size and number of areas.
	}
	\label{fig:metric}
\end{figure}

\begin{figure}[!h]
    \centering
    \includegraphics[width=\imgfactor\columnwidth]{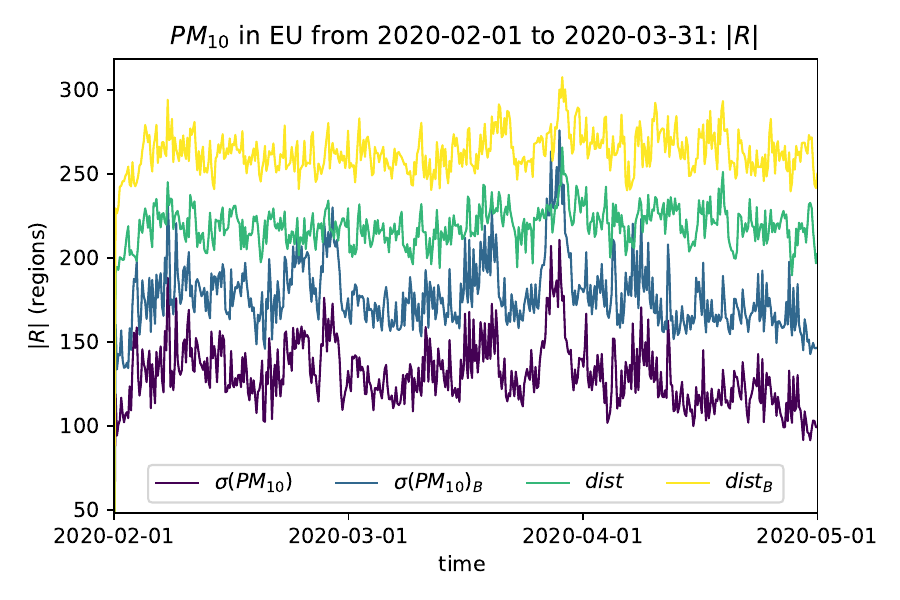}
    \includegraphics[width=\imgfactor\columnwidth]{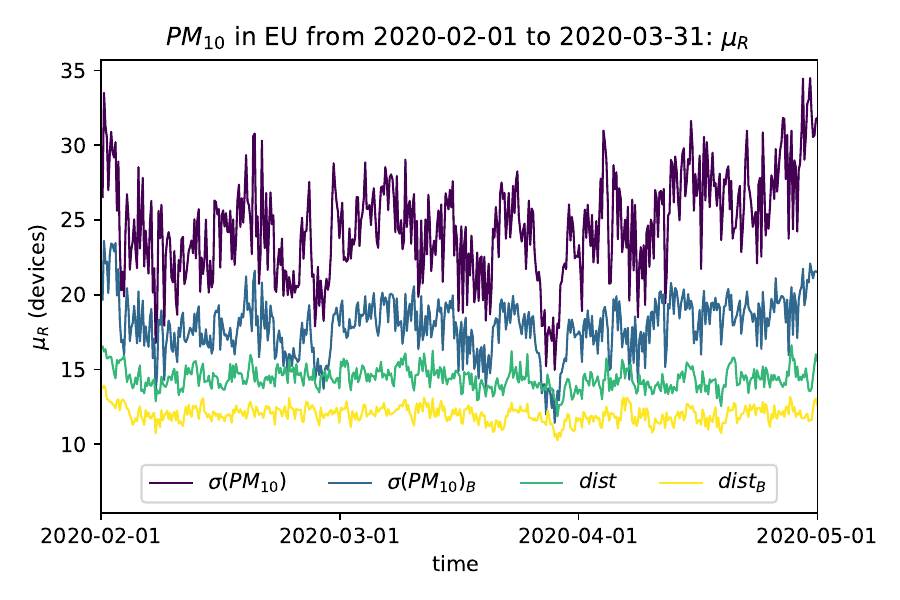}
    \includegraphics[width=\imgfactor\columnwidth]{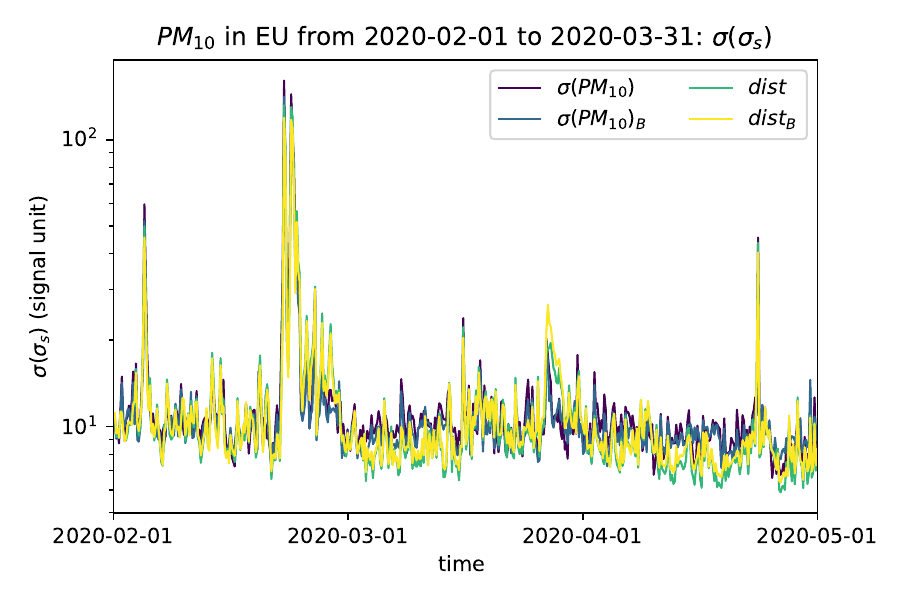}
    \includegraphics[width=\imgfactor\columnwidth]{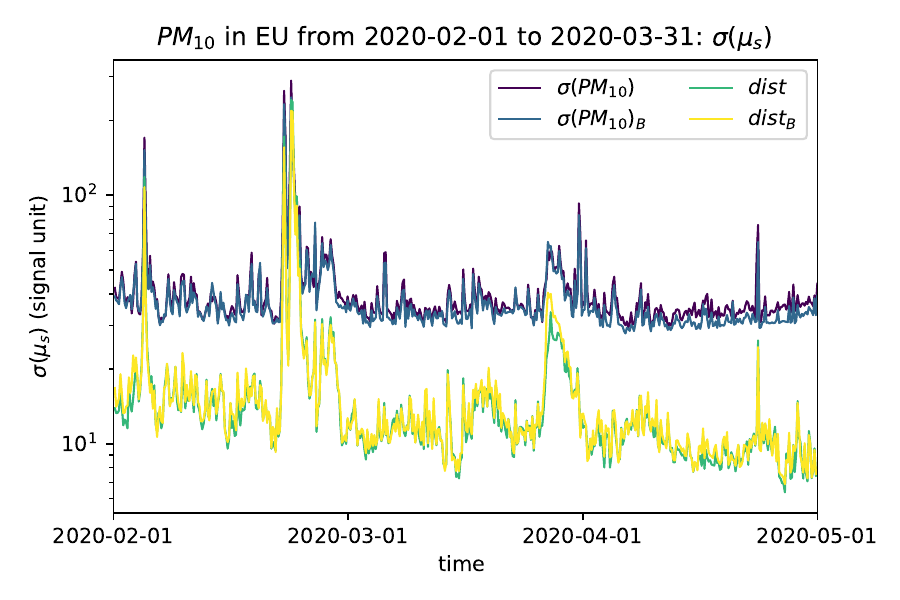}
    \caption{
        \revA{Region count (top left),
        mean region size (top right),
        intra-region error (bottom left), and
        inter-region difference (bottom right)
        with different metrics
        for the realistic experiment.
        When fed with a metric that considers both distance and error appropriately,
        the algorithm is able to create regions in a way that is both efficient and accurate.
        This is especially true for the ability to maximise the difference among different regions (bottom right).
        The algorithm is also robust to the presence of arbitrary-set limits,
        such as the country borders:
        although introducing them as an impassable barrier generates (with the same metric) a higher region count with a smaller mean size (top charts),
        the quality of the solution is only marginally affected (bottom charts).
        }
    }
    \label{fig:charts-real}
\end{figure}

\subsection{Results}
\label{ssec:eval-results}

\revA{The full analysis counts 637 charts; the interested reader can check them out in the experiment repository.
In this manuscript, we show the most relevant ones, that we believe help to shed light on the behaviour of the proposed algorithm.
}

In \Cref{fig:region-count-deployment},
we show that our implementation stabilises,
since after a short transition all values become stable.
Obviously, in the dynamic case, these transitions are present throughout the experiment.
As expected,
the aggregate sampler defines a number of regions that differs depending on the underlying phenomenon under observation.
From \Cref{fig:intra-inter-deployment} and \Cref{fig:intra-deployment},
we notice that the system indeed tries to maximise inter-region differences
and minimise intra-region differences,
thus effectively addressing the trade-off between \emph{high information (entropy)} and \emph{error-controlled upscaling} (as per \autoref{def:optsampling}).
Finally, \Cref{fig:leader-selection} and \Cref{fig:metric} show how the algorithm reacts to changing parameters.
As expected,
while modifying the leader selection policy has minimal impact on the behaviour of the system,
changing the error-distance metric greatly affects its behaviour.
In all cases, we notice that the driver signal with higher information entropy
(uniform) generates a larger number of smaller regions than all other signals,
while the one with the lowest information entropy
(constant)
always produces few (usually one) large regions.
The reason is that the leader selection impacts the originating point of a region,
but it is its expansion (driven by the metric) that ultimately determines its extension and shape.

\revA{These findings, coming from the analysis of the results produced by the synthetic environments,
are confirmed when real-world data is used instead.
In \Cref{fig:charts-real}, indeed, we see that the algorithm is very sensitive to the choice of the metric
(as expected);
yet, it is pretty robust to the presence of arbitrarily-set spatial limits such as country borders:
although they inevitably lead to more regions of a smaller size,
the key performance indicators (high inter-regional difference and low intra-regional difference)
remain pretty much unchanged.
}

\section{Conclusions and Future Work}
\label{s:conc}

In this paper, 
we tackled the problem of defining an \sampler{} 
sensitive to the spatial dynamics of the phenomenon under observation.
In particular,
we wanted to
minimise the sampling error
(minimum when all available sampling devices are used)
while also minimising the regions count---two contrasting goals.

We formalised the problem  
within the framework of event structures and field-based coordination, 
suitable to represent situated, large-scale, and dynamic computations.
We thus designed a spatial adaptive aggregate sampler
based on a leader election strategy 
that dynamically creates and grows/shrinks sampling clusters (or regions)
based on
 two main control ``knobs'':
the error-distance metric and the leader strength.
\addrev{
We proved that the proposed algorithm is self-stabilising and enjoys a local optimality property.
}
Through simulation, the proposed algorithm is shown to satisfy the mentioned trade-off. 

As measuring performance and efficiency of such an adaptive algorithm is far from trivial, 
we exploited several metrics to validate intended behaviour.
However, as a follow-up work we would like to synthesize a single indicator able to measure both accuracy and efficiency, 
using information theory such as those derived from entropy 
(e.g.\ mutual information).
Also, we are analysing openly available air pollution datasets to design new simulations based on real-world data,
to better emphasise the impact that our aggregate sampler could have for policy making based on spatial phenomena.
Finally, future work will be devoted to investigating
how space-fluid sampling can integrate with time-fluid aggregate computations~\cite{lmcs-timefluid}.

\section*{Acknowledgment}
This work has been supported by the MIUR PRIN 2017 Project
``Fluidware'' (N. 2017KRC7KT) and the EU FSE PON R\&I 2014-2020.

\bibliographystyle{alphaurl}
\bibliography{biblio}

\end{document}